\documentclass[journal,transvt]{IEEEtran}

\usepackage{amsfonts}
\usepackage{dsfont}
\usepackage{setspace}
\usepackage{color}
\usepackage{amssymb}
\usepackage{cite}
\usepackage[cmex10]{amsmath}
\usepackage{algorithm}
\usepackage{algorithmic} 
\usepackage{array} 
\usepackage{mathrsfs}
\usepackage{graphicx}
\usepackage{latexsym}
\usepackage{amscd}
\usepackage{amsfonts}
\usepackage{subfigure}
\usepackage{amsmath,amscd,amssymb,verbatim}
\usepackage{graphics}
\usepackage{amsthm}
\usepackage[T1]{fontenc}
\usepackage[utf8]{inputenc}
\usepackage{authblk}
\usepackage{optidef}
\usepackage{amsmath,amsthm}
 
\newtheorem{prop}{Proposition}

\newtheorem{lem}{Lemma}

\begin{document}

\title{Optimal Resource Allocation in Ground Wireless Networks Supporting Unmanned Aerial Vehicle Transmissions
\thanks{
Y. Hu, G. Sun and A. Schmeink are with the Informationstheorie und Systematischer Entwurf von Kommunikationssystemen (ISEK Research Area), RWTH Aachen University,
52074 Aachen, Germany (e-mail: \{hu,sun,schmeink\}@isek.rwth-aachen.de).}
\thanks{G. Zhang is with the School of Communications and Information Engineering, Xi’an University of Posts and Telecommunications, Xi’an 710121,
China (e-mail: zhangghcast@163.com).}
\thanks{M. C. Gursoy is with Department of Electrical
Engineering and Computer Science, Syracuse University, Syracuse, NY 13244, US (email: mcgursoy@syr.edu).}
}

\author[$\dag$]{
Yulin Hu,~\IEEEmembership{Senior Member,~IEEE,} Guodong Sun, Guohua Zhang, M. Cenk Gursoy,~\IEEEmembership{Senior Member,~IEEE,} and Anke Schmeink,~\IEEEmembership{Senior Member,~IEEE} 
\vspace*{-.35cm}  
}

\maketitle

\begin{abstract}
We consider a fully-loaded  ground  wireless network supporting unmanned aerial vehicle (UAV) transmission services.
To enable the overload transmissions to a ground user (GU) and a UAV, two  transmission schemes are employed, namely non-orthogonal multiple access (NOMA) and relaying, depending on whether or not the GU and UAV are served simultaneously.
Under the assumption of the system operating with infinite blocklength (IBL) codes, the IBL throughputs of both the GU and the UAV are derived under the two schemes.  
More importantly,  we also consider the scenario in which data packets are transmitted via finite blocklength (FBL) codes, i.e., data transmission  to both the UAV and the GU   is performed under low-latency and high reliability constraints. In this setting, the FBL throughputs are characterized again considering the two schemes of NOMA and relaying.
Following the IBL and FBL throughput characterizations, optimal resource allocation designs are subsequently proposed to maximize the UAV throughput while guaranteeing the throughput of the cellular user. 
 Moreover, we prove that   the relaying scheme is able to provide transmission service to the UAV while improving the GU's performance, {and that} the relaying scheme potentially offers       a higher throughput to the UAV in the FBL regime than in the IBL regime.  {On the other hand, the NOMA scheme  provides a   higher UAV throughput (than relaying) by slightly sacrificing the GU's performance.}  
\end{abstract}


\vspace*{-.2cm}  
\section{Introduction}

\IEEEPARstart{T}he past few years have witnessed a tremendous increase in the deployment and use    of unmanned aerial vehicles (UAVs) in both commercial and civilian applications, such as for aerial surveillance, traffic control, photography, package delivery, and so on~\cite{zeng2016wireless,van2016lte,seo2018drone,duque2018synthesis}. 
Hence, the ground wireless networks are expected to not only serve ground users (GU) but also  provide ubiquitous communication services for the UAVs, i.e., support the communication between the UAVs and their users or controllers~\cite{khosravi2018performance,gapeyenko2018flexible}.
Consequently, ground networks supporting UAV communications have relatively more workload. 

Non-orthogonal multiple access (NOMA) and relaying are known as promising techniques enhancing the performance of overloaded networks.
 By applying the successive interference cancellation (SIC) technique, NOMA enables two (or more) users to transmit data packets in a single time-frequency resource block by exploiting differences in channel conditions~\cite{saito2013system}. 
Authors in~\cite{ding2016impact} address the impact of user pairing on the NOMA system. In~\cite{ali2016dynamic}, the sum-throughput achieved with NOMA is investigated. 
Recently, for a UAV-assisted NOMA network where the UAV serves as  an access-point,     the outage probability of   GUs~\cite{sharma2017uav} and the UAV trajectory designs~\cite{Zhao_2018} are addressed, respectively.
In addition, cooperative relaying has been also shown as an efficient way to improve the system  performance~\cite{zeng2016throughput,johansen2014unmanned} by letting one user (with a better link to the access point) perform as a relay to assist the other weak-link user(s). 
For a UAV network,  it has proposed in~\cite{zeng2016throughput} to apply  the UAV as a relay and study the throughput maximization of the system. In addition, experiments on such     UAV-relaying ad hoc networks are     conducted in~\cite{johansen2014unmanned}. 
{From a resource allocation perspective, the authors in \cite{tang2018novel}     maximize  the precision of a UAV-aided recommendation system by optimally allocating  the resource  of computing, delay, and traffic resource. 
 The allocation of tours of targets to vehicles is investigated in \cite{rathinam2007resource} for a multi-UAV system 
  to minimize the sum of the flying distances of UAVs. The     frequency allocation problem is addressed   in~\cite{tang2018ac}    for a UAV-assisted D2D network.}
However, fundamental characterizations and optimal designs for the network of a UAV  in pair with a GU are yet to be identified under both the NOMA and relaying schemes. 

Moreover, all the above studies are conducted under the assumption of transmitting arbitrarily reliably at Shannon's capacity. In other words, these results are only accurate for the network with significantly long blocklengths.
Note that both UAV and  GU can be deployed to support latency-critical applications with high reliability requirements,  e.g., remote control. 
Due to these critical latency and reliability requirements, data transmissions are realized by codes with short blocklengths, i.e.,  the  networks  operate  in  the  so-called  finite  blocklength (FBL) regime in which the transmissions are no longer arbitrarily reliable~\cite{she2017radio}. 
 To tackle this problem,   FBL information theoretic   bounds are developed in \cite{polyanskiy2010channel} for an additive Gaussian noise (AWGN) in a single hop transmission. In addition, the model has been extended to quasi-static fading channels~\cite{Xu_2016,she2017cross},  cooperative relaying networks~\cite{Hu_2016} and NOMA networks~\cite{Hu_2017_PIMRC,NOMA_FBL_2018}. 
More recently, a joint blocklength and UAV location optimization is studied in~\cite{Pan_2019} for a FBL network where the UAV is considered as  a flying base-station/relay. 
{In addition, the authors in \cite{she2018UAV} consider  a UAV network including   a ground access point (AP) and  UAV-users,  and maximize  the available range of the AP to support FBL  communications for these UAV-users.}
 However,   note that  with  a high probability of having a line-of-sight (LoS) link, a UAV-user generally has   strong channels   from the ground AP and to GUs. On the one hand, significant channel quality difference generally exists between the channels from AP to the UAV-user and     to a GU, which motivates one to apply      NOMA between the two users experiencing the channel difference. 
  On the other hand, as the link from the UAV-user to the GU is likely to be strong, this indicates a potential performance improvement (in both IBL and FBL regimes) for the GU when the UAV-user performs as a relay to assist the data transmission to the GU.   Hence, it is essential to study and optimize the performance of such pair of  users.
{To the best of our knowledge, both the IBL and the FBL performances of a UAV-user  paired   with a GU under either the NOMA scheme or the relaying scheme have yet to be     characterized and maximized via optimal resource allocation.}

 In this work, we consider a   ground wireless network supporting UAV transmissions, where a UAV and a GU are served under either the NOMA or relaying scheme.  The throughputs of UAV and GU are characterized in the IBL regime  and the FBL regime, addressing the latency-non-sensitive applications and low-latency high-reliability applications, respectively. 
 {Following   these characterizations,   we propose optimal resource allocations to maximize the UAV throughput while guaranteeing the GU's transmission quality.}
The contributions of this paper can be further detailed as follows:
\leftmargini=4mm
\begin{itemize}\itemsep=0.7pt
  \item The throughputs of UAV and GU are characterized in both the IBL and FBL regimes. 
  \item  Following   the throughput characterizations,    optimal resource allocation designs are determined  under   the NOMA  and relaying schemes, when the objective is to maximize the UAV throughput while guaranteeing the GU's transmission quality. The designs are addressed  in the IBL regime and the FBL regime, respectively.  
  \item We  further address the feasibility of the considered optimization problems. 
 {In particular, we prove that the  relaying is able to create a win-win situation, i.e.,   provides transmission service to the UAV while improving the GU's performance.  Moreover, we prove that the relaying scheme potentially offers       a higher throughput to the UAV in the FBL regime than in the IBL regime, as in the FBL regime it likely requires less resources to guarantee the GU's transmission quality and therefore     allocates more resources to   the UAV.}
\end{itemize}

The remainder of the paper is organized as follows:
In Section II, we describe the system model, review the FBL performance metrics and provide the problem statement.
In Section III, throughputs of the considered network are characterized under the  NOMA and relaying schemes, respectively. 
Subsequently in Section IV, we  analyze the optimal resource allocation strategies, which maximize the UAV throughput while guaranteeing the GU's transmission quality. The feasibility of the   resource allocation problems are further discussed in Section V. 
We provide our simulation results in Section VI 
and finally conclude the paper in Section VII.

\section{Preliminaries}

In this section, we first briefly describe the considered scenario of a UAV connected ground wireless network. Subsequently, we introduce the channel model representing the wireless link characteristics. Finally, we review the FBL performance model of a single link transmission in comparison to the IBL performance model.  

\subsection{System Description}

We consider a ground wireless network with an AP serving its own GUs. At the same time, a UAV is additionally required to be served beside the GUs. In particular, due to the resource (e.g.,  frequency, power)   limitation in the   network, 
the service to the UAV is     provided by using the existing resources 
i.e.,  without incurring additional costs in terms of more frequency bands or increased transmit power in the network. To satisfy this condition, one practical approach     is to select a GU and use the frequency   band allocated to this GU to serve both the GU and the UAV jointly, as shown in Fig.\ref{UAVcell}. At the same time, the selected GU's transmission quality (with respect to throughput and reliability) should be guaranteed (in comparison to being solely served) when applying such joint transmission schemes\footnote{{Generally, if we select the GU which has a  poor channel to the AP,     its transmission performance via the AP-GU link (being solely served before paring with the UAV) is   low. Hence a similar transmission performance of this GU is relatively easier to be guaranteed by applying joint transmission schemes to the GU-UAV pair.}}.

\begin{figure}[h]
	\centering
	\includegraphics[width = 0.7\linewidth,trim = 0 10 10 35]{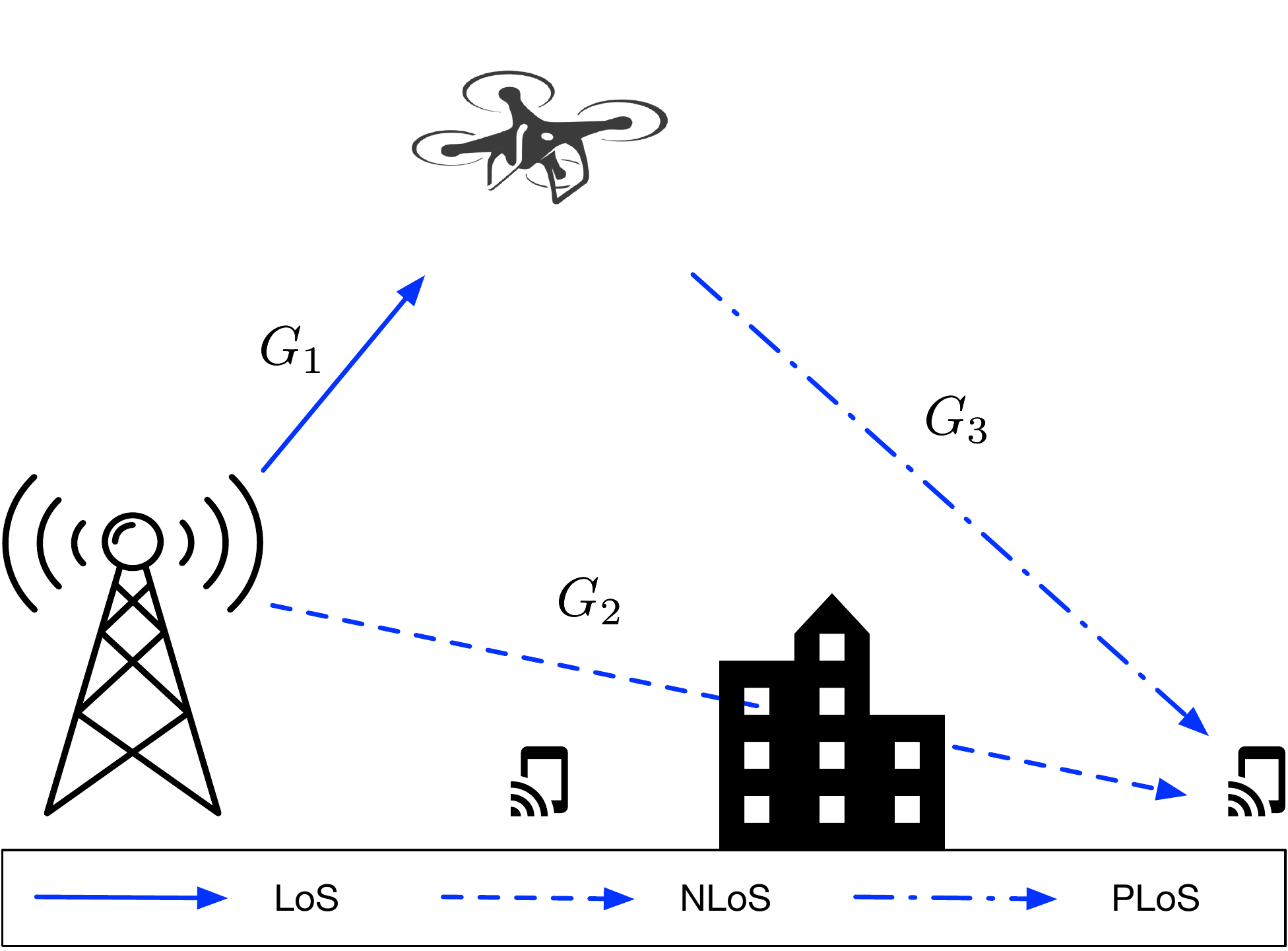}
	\caption{An example of the considered network, where the AP-UAV link has the line of sight (LoS), the AP-GU link has no LoS (NLoS), and the UAV-GU link is possible to have the LoS (PLoS).}
	\label{UAVcell}
\end{figure}

We assume that the system operates in a time-slotted fashion where time is divided into frames with a length of $M$ symbols. 
The gains of the channels from the AP to the UAV, from the AP to the GU, and from the UAV to the GU are denoted by $G_1$, $G_2$, and $G_3$, respectively.
Then, the   channel gains are modeled by $G_i = z_i 10^{-L_i}$, $i=1,2,3$, where  $L_i$, $i=1,2,3$, represent the path losses (in dB) in   these links. 
In addition,  $z_i$, $i=1,2,3$,    denote the gains due to the random channel fading.  
In particular, we assume  that channels    experience block-fading, i.e., $z_i$, $i=1,2,3$, are constant within a frame and vary from one frame to the next.

\subsection{IBL Regime vs. FBL Regime}

First of all, it should be pointed out that for a practical communication system, the blocklength is definitely not infinite.
Hence, when we study the performance of such a system in the IBL regime, this only indicates that the impact of the limited length of coding blocks is ignored in the analysis.
In particular, the analysis in the IBL regime follows the assumption which is only true when the blocklength is infinitely long: a packet is
assumed to be decoded with arbitrarily small error probability given that the coding rate is lower than the Shannon capacity.

In an AWGN channel, the Shannon capacity, which quantifies     throughput  in the IBL regime,  is given in bits/symbol~by
\begin{equation}
	C =\log_2(1+\gamma)\, ,
\end{equation}
where $\gamma$ is the signal to     noise ratio (SNR) or signal to interference plus noise ratio (SINR)       at the receiver. 

Under the more practical assumption that the coding blocklength is finite, i.e., the system operates in the FBL regime, the decoding error may occur even when   the transmission coding rate is set to be lower than the Shannon capacity.
In  an AWGN channel, the normal approximation of the coding rate $r$ (in bits per channel use) of transmission is derived in~\cite{polyanskiy2010channel}. 
 Later on, this approximation is improved to be more tighter in~\cite{Third_Order}, where the third-order term of  the normal approximation is   derived. 
 In particular, 
with    blocklength $m$, block error probability~$\epsilon$, and SNR (or SINR) $\gamma$, the coding rate $r$ is given by \cite{Third_Order} 
\begin{equation} \label{Rfunciton}
\begin{split}
	r &= \mathcal{R}(m, \gamma, \epsilon) \\ 
	&= \log_2(1+\gamma) - \log_2e\sqrt{\frac{\gamma(\gamma+2)}{(\gamma+1)^2m}}Q^{-1}(\epsilon) + \frac{\log m}{m}  +\frac{o(1)}{l}  \,, 
\end{split}	
\end{equation}
where $Q^{-1}(x)$ is the inverse $Q$-function and  the $Q$-function is given by $Q(x) = \int_0^{\infty} \frac{1}{\sqrt{2\pi}}e^{-t^2/2}dt$.
On the other hand, when the coding rate is fixed as $r$, the error probability of the transmission is given by
\begin{equation} \label{Rfunciton}
\epsilon = \mathcal{P}(m, \gamma, r)  = Q\left( {\frac{{{{\log }_2}(1 + \gamma ) - r + \frac{{\log m}}{m} -  \frac{o(1)}{l} }}{{{{\log }_2}e\sqrt {\frac{{\gamma (\gamma  + 2)}}{{{{(\gamma  + 1)}^2}m}}} }}} \right).
\end{equation}
According to the characterizations in~\cite{polyanskiy2010channel},  the achievable coding/data rate in the FBL regime increases in the blocklength. In particular, there exists a performance gap between the Shannon capacity and FBL achievable rate.  This performance gap in a single-hop system is numerically shown in Fig.~\ref{Performance-gap}. 
{\begin{figure}[!b]
\begin{center}
\includegraphics[width=0.7\linewidth, trim=30 15 30 10]{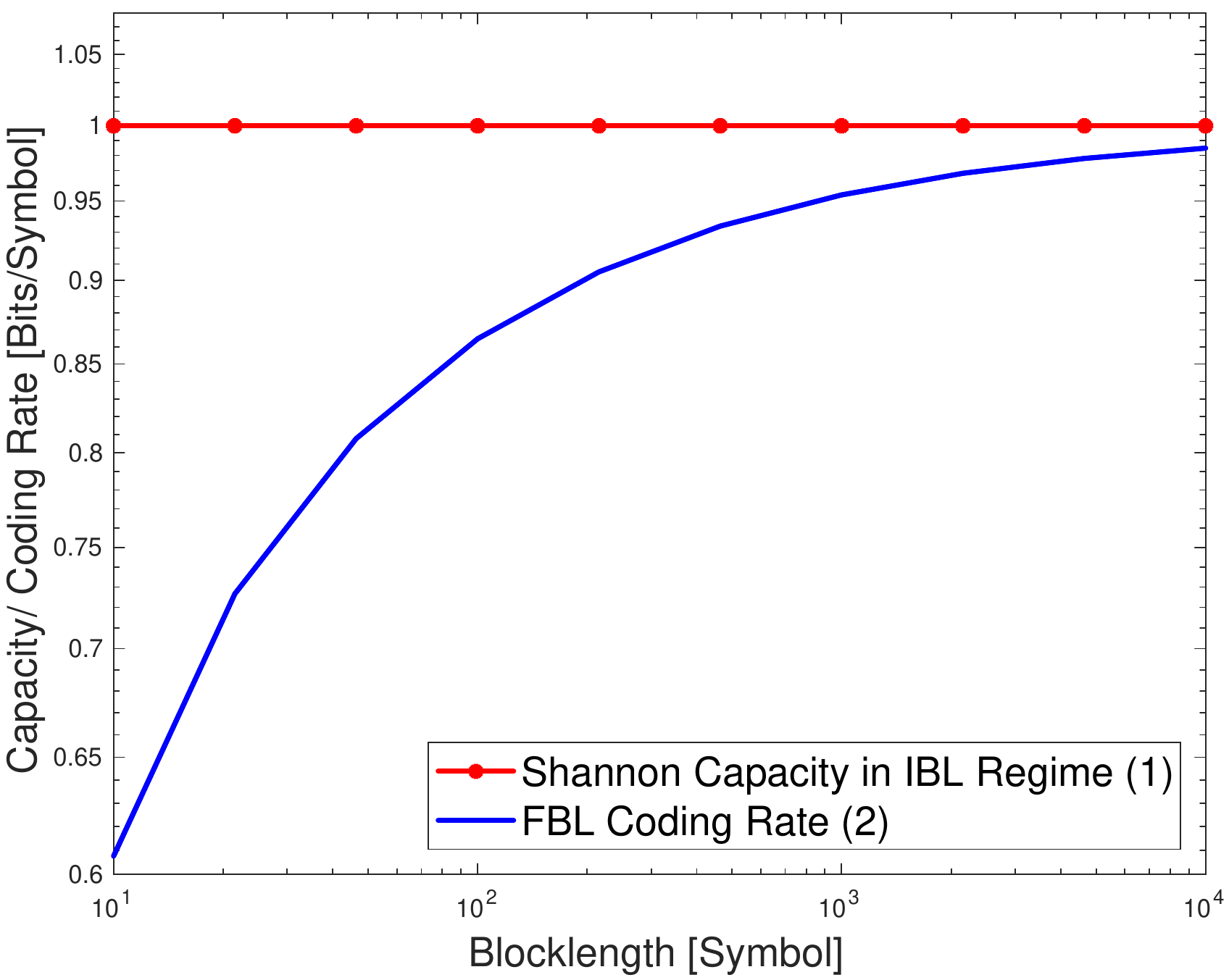}
\end{center}
\caption{Performance gap between the IBL and FBL regimes in a single-hop system with a static channel. The FBL (achievable) rate is calculated according to (2) while  the target error probability is set as $10^{-4}$. }
 \vspace*{-5pt}
\label{Performance-gap}
\end{figure}}
The figure illustrates that the gap is more significant for short blocklengths, and the FBL performance degrades significantly as the blocklength is decreased. 

\subsection{Problem Statements}

Our main focus in this work is to investigate the performance of the UAV-connected wireless system. {We propose to apply two transmission schemes and evaluate the throughput to the UAV.} In particular, we address the following fundamental problems: When comparing NOMA and relaying, which one is more preferred under which condition? How to provide an optimal throughput to the UAV (under both the NOMA and relaying schemes) while guaranteeing the GU's transmission quality by applying optimal resource allocation?
What are the performance   differences in the considered network under the IBL and FBL assumptions?  How different the optimal resource allocation solutions are in the FBL regime in comparison to the IBL regime?

\section{Transmission Schemes and Throughput Characterizations}

In this section, we present the NOMA and relaying schemes and characterize the corresponding throughputs in the links    to the UAV and the GU. The characterizations of the two schemes are provided in both the IBL and FBL regimes.

\subsection{IBL Regime}
%
If the AP serves only its GU and not the UAV,  the   throughput of the GU is given by
\begin{equation}\label{mu_0_IBL}
\mu^{\rm{IBL}}_0 = M \log_2 \left(1+\frac{p_0G_2}{\sigma^2}\right).
\end{equation}

If, on the other hand,  the UAV is served jointly with the GU, the system behavior changes. In the following, we address the throughputs   (including both the UAV and GU) considering the NOMA and relaying schemes, respectively.  

\subsubsection{NOMA} 
 In the NOMA scheme,   AP transmits two signals $\mathbf{x}_{\text{uav}}$ and $\mathbf{x}_{\text{gu}}$ within the same coding length of $M$ symbols via the same frequency band to the UAV and GU simultaneously.
 Denote by $p_1$ and $p_2$ the per symbol transmit power of the two transmitted signals, i.e., $p_1=\mathbb{E} \{\|\mathbf{x}_{\text{uav}}\|^2\}/M$ and $p_2=\mathbb{E} \{\|\mathbf{x}_{\text{gu}}\|^2\}/M$.
 Note that the AP-UAV link is with high probability to have line of sight (LoS) while the link from the AP to the selected GU (that is paired with the UAV) can experience more severe fading (e.g., due to blockages and multipath fading).
 In other words,  the quality of the AP-UAV link  is higher than that of the AP-GU link, i.e., $G_1>G_2$.  Hence, without loss of generality, in the NOMA process we assume $p_2 > p_1$  and let the UAV apply SIC to cancel the inference of $\mathbf{x}_{\text{gu}}$. 

The received signal at the UAV is given by 
 \begin{equation}\label{Received_s_uav}
 y_{1} = \sqrt{p_1G_1}\mathbf{x}_{\text{uav}} + \sqrt{p_2G_1}\mathbf{x}_{\text{gu}} + n \, ,
 \end{equation}
 where $n$ denotes the Gaussian noise with power $\sigma^2$.
Note that the interference signal needs to  be successfully cancelled, which requires $\frac{p_2G_1}{p_1G_1+\sigma^2}>\frac{p_2G_2}{p_1G_2+\sigma^2}$.
Moreover, the IBL throughput of the UAV under the NOMA scheme is given by
\begin{equation}\label{muNOMAuavIBL}
\mu^{\rm{IBL}}_{\text{N},1} = M \log_2\left(1+\frac{p_1G_1}{\sigma^2}\right) \, .
\end{equation}
   
On the other hand, the signal received by the GU is    
   \begin{equation}\label{Received_s_ue}
 y_2 = \sqrt{p_1G_2}\mathbf{x}_{\text{uav}}+\sqrt{p_2G_2}\mathbf{x}_{\text{gu}}+n \, .
 \end{equation}
 Decoding the signal $\mathbf{x}_{\text{gu}}$ while treating  $\mathbf{x}_{\text{uav}}$ as   interference, the GU's throughput is given by
 \begin{equation}\label{muNOMAueIBL}
 	\mu^{\rm{IBL}}_{\text{N},2}= M \log_2\left(1+\frac{p_2G_2}{p_1G_2+\sigma^2}\right) \, .
 \end{equation}
   
\subsubsection{Relaying}
Under the relaying scheme, the UAV becomes both a data receiver and a relay    to forward the GU's data packet. 
In this scheme, the total length $M$ of a frame is further divided into two phases with lengths   $m_1$ and $m_2$, satisfying $m_1 +m_2 =M$. In the first phase, the AP transmits a large data packet (containing the messages intended for  both the UAV and GU)  to the UAV. 
If the UAV decodes the large data packet correctly and recovers both messages, it forwards the GU's data in the second phase. 
Thus, the IBL throughput of  the GU under the relaying scheme is
\begin{equation}\label{RalayIBLue}
 	\mu^{\rm{IBL}}_{\text{R},2}= m_2 \log_2\left(1+\frac{p_2G_3}{\sigma^2}\right) \, .
\end{equation}
Note that the large data transmitted to the UAV includes both the UAV's data and the GU's data. 
Hence, under the relaying scheme the throughput via the link to the UAV is given by $m_1 \log_2\left(1+\frac{p_1G_1}{\sigma^2}\right)$, which equals the sum of  the UAV's IBL throughput   $\mu^{\rm{IBL}}_{\text{R},1}$ and the GU's throughput $\mu^{\rm{IBL}}_{\text{R},2}$. Therefore, the IBL throughput of the UAV can be expressed as
  
\begin{equation}\label{RelayIBLuav}
	\mu^{\rm{IBL}}_{\text{R},1}= m_1 \log_2\left(1+\frac{p_1G_1}{\sigma^2}\right) - \mu^{\rm{IBL}}_{\text{R},2} \, .
\end{equation}

\subsection{FBL Regime}
Note that in the FBL regime, transmissions are not arbitrarily reliable anymore. 
We denote by~$\nu_1$ and  $\nu_2$ the reliability requirements of the transmissions to the UAV and the GU, respectively. 
When only the GU is served by the BS, the total error probability is exactly the decoding error probability at the GU. 
According to~\eqref{Rfunciton}, when the AP only serves the GU, the FBL throughput is given by
\begin{equation}\label{FBL_GU_requirement}
	\mu^{\rm{FBL}}_0 = M\cdot \mathcal {R} \left(M, \frac{p_0G_2}{\sigma^2}, 
	\nu_2 \right) \cdot (1- \nu_2
	) \, ,
\end{equation}

In the following, we discuss the FBL throughput performance of the NOMA and relaying schemes while guaranteeing the GU's reliability requirement. 

\subsubsection{NOMA}

In the NOMA scheme, the UE decodes the received signal given in~\eqref{Received_s_ue} by treating $\mathbf{x}_{\text{uav}}$ as   interference. 
Hence, the GU's the throughput is given by
\begin{equation}\label{FBLmuUserNOMA}
	\mu^{\rm{FBL}}_{\text{N},2} = M \cdot r_2 \cdot (1-\nu_2) \, ,
\end{equation}
where $r_2=\mathcal{R}\left(M,\frac{p_2G_2}{p_1G_2+\sigma^2},\nu_2\right)$ is the coding rate  of the data to the GU.

On the other hand,
the decoding process at the UAV consists of two steps, i.e.,  UAV first applies SIC and subsequently decodes its own signal. Hence, both the SIC errors and the  decoding errors of the UAV's own signal contribute to the  total error probability.
Denote by $\epsilon_{{\text{SIC}}}$ and $\epsilon_{{\text{N},1}^{\ast}}$   
the error probabilities of SIC and decoding the UAV's data packet at the UAV. 
Then,     the total error probability of the transmission  to the UAV is given by $1-(1-\epsilon_{{\text{SIC}}})(1-\epsilon_{{\text{N},1}^{\ast}}) $.
Note that the error probability of this SIC process is given by
\begin{equation}
	\epsilon_{\text{SIC}} = \mathcal{P}\left(M,\frac{p_2G_1}{p_1G_1+\sigma^2}, r_2\right) \, .
\end{equation}
Hence, to satisfy the reliability requirement, i.e.,  $1-(1-\epsilon_{{\text{SIC}}})(1-\epsilon_{{\text{N},1}^{\ast}}) = \nu_1$, the target error probability of the UAV's data transmission is required to be set to\footnote{It should be pointed out that the decoding error probability is set as large as possible to ensure the loosest   constraint so that the maximum throughput could be obtained.} $\epsilon_{{\text{N},1}^{\ast}}=1- \frac{1-\nu_1}{1-\epsilon_{\text{SIC}}}$.

After successfully applying the SIC, the UAV subsequently decodes it own data. 
The FBL throughput of the UAV in the NOMA scheme is given by 
\begin{equation}\label{muNOMAuavFBL}
	\begin{split}
	\mu^{\rm{FBL}}_{\text{N},1}  &= M \cdot  r_1 \cdot (1-\epsilon_{{\text{SIC}}})(1-\epsilon_{{\text{N},1}^{\ast}}) \\ 
	& = M \cdot  r_1 \cdot (1-\nu_1),   
	\end{split}
\end{equation}
where $r_1 = \mathcal{R}(M, \frac{p_1G_1}{\sigma^2}, \epsilon_{{\text{N},1}^{\ast}})$ is the coding rate of the data packet to the UAV.

\subsubsection{Relaying}

Under the relaying scheme, the AP first transmits a large data packet in first phase to the UAV. 
The decoding error probability of this transmission is exactly the total error probability of  the UAV's data transmission, i.e., it is required to satisfy the reliability constraint~$\nu_1$. 
Hence, the FBL throughput of the transmission to the UAV, including the data for both UAV and GU,  is given by
\begin{equation}
\mu^{\rm{FBL}}_{\text{R},1+2} = m_1 \cdot \mathcal {R}\left(m_1, \frac{p_1G_1}{\sigma^2}, \nu_1\right)\cdot (1-\nu_1)\, .
\end{equation}

As long as the UAV's decoding process is successful, the UAV forwards the GU's data packet in the second phase. 
Denote by the $\epsilon_{{\text{R},2}^{\ast}}$ the target decoding error probability at the GU. According to the total reliability constraint of GU $\nu_2$, it holds that $\nu_2=1-(1-\nu_1)(1-\epsilon_{{\text{R},2}^{\ast}})$. Hence, we have $\epsilon_{{\text{R},2}^{\ast}} =1- \frac{1-\nu_2}{1-\nu_1}$.
Then, the GU's FBL throughput is given by  
\begin{equation}
\begin{split}
    	\mu^{\rm{FBL}}_{\text{R},2} &= m_2   \mathcal {R}\left(m_2, \frac{P_2G_3}{\sigma^2}, \epsilon_{{\text{R},2}^{\ast}}\right)\cdot(1-\nu_1)(1-\epsilon_{{\text{R},2}^{\ast}})\\
    	&=m_2   \mathcal {R}\left(m_2, \frac{P_2G_3}{\sigma^2}, \epsilon_{{\text{R},2}^{\ast}}\right)\cdot(1-\nu_2) \, .
\end{split}
\end{equation}
Finally, the UAV's throughput can be obtained as
\begin{equation}\label{FBLRelay}
	\mu^{\rm{FBL}}_{\text{R},1}  =  m_1 \cdot \mathcal {R}\left(m_1, \frac{p_1G_1}{\sigma^2}, \nu_1\right)\cdot (1-\nu_1) - \mu^{\rm{FBL}}_{\text{R},2} \, .
\end{equation}

So far, we have characterized the throughput performance of the considered NOMA and relaying schemes in both the IBL and FBL regimes. Following these characterizations, we provide the corresponding optimal resource allocation designs in the next section.

\section{Resource Allocation for Maximal UAV Throughput}

In this section, by applying resource allocation, we investigate the optimal UAV throughput  under both the NOMA and relaying schemes.  
{Our objective is to maximize the throughput without using any additional resources (e.g., in terms of transmission blocklength and  transmit power), 
while guaranteeing the GU's transmission requirements including throughput and reliability (in the FBL case).} In particular, we seek to guarantee the GU a weighted throughput level of $\beta \mu_0$, where $\mu_0$ is the GU throughput if the UAV is not jointly served, and $\beta >0$ is the minimum rate coefficient, which is a key factor  that indicates the level of service guarantee. For example, $\beta<1$ indicates that 
only partial throughput can be provided to the GU in comparison to the   case in which the UAV is not serviced. On the other hand, $\beta>1$ represents the scenario in which GU performance is required to be better than the non-UAV-connected case.  

\subsection{UAV Throughput Maximization under the NOMA Scheme }

According to \eqref{muNOMAuavIBL} and \eqref{muNOMAuavFBL}, 
the UAV throughput in the NOMA scheme can be maximized by choosing an optimal $p_1$. 
The general optimization problem of the NOMA scheme (in both the IBL and FBL regimes) is formulated as:
\begin{maxi!}|s|[2]
{p_1}{\mu_{\text{N},1}}
{\label{NOMA_PA_or}}{} 
\addConstraint{\mu_{\text{N},2}}{\geq \beta \mu_{0} \label{NOMAg}}
\addConstraint{p_1+p_2} {\leq p_0 },{\quad p_1,p_2 \in R^+ \, . \label{NOMAp}}
\end{maxi!}
We have the following proposition for the above problem 
\begin{prop}\label{strict_bound}
	The optimal throughput offered to the UAV is obtained when the inequality constraints in~\eqref{NOMAg} and~\eqref{NOMAp} are satisfied with equality.
\end{prop}
\begin{proof}
 We prove the result by contradiction. Regarding the sum power constraint in \eqref{NOMAp},  we first assume that there exists an optimal solution $(p_1', p_2')$ satisfying the constraints with strict inequality, i.e. $p_0 - p_1'+p_2' = p_n'>0$. Hence the optimal throughput offered to the UAV and the GU are $\mu_{uav}$ and $\mu_{gu}$. On the other hand, we could further allocate the power left to UAV $p_1''$ and GU $p_2''$ in proportion to $p_1'$ and $p_2'$.Thus the data rate to the GU is $c_{gu} = \log_2\left(1+\frac{(p_2'+p_2'')G_2}{(p_1'+p_1''+ p_2'+p_2'')G_2+\sigma^2}\right)$. Since $\frac{(p_2'+p_2'')G_2}{(p_1'+p_1''+ p_2'+p_2'')G_2+\sigma^2} > \frac{p_2'G_2}{(p_1'+p_2')G_2 + \frac{\sigma^2 p_1'}{p_1'+p_1''}}$, the GU could get more throughput. At the same time, the data rate of UAV is $c_{uav} = \log_2\left(1+\frac{(p_1'+p_1'')G1}{\sigma^2}\right)$, and thus the throughput to the UAV is also higher. Therefore, the assumption that $(p_1', p_2')$ is the optimal solution to the problem is violated. i.e., the optimal power allocation satisfies \eqref{NOMAp} with equality.  
  Similarly, the strict equality   of~\eqref{NOMAg} and guaranteeing constraint can be proved. 
\end{proof}

According to Proposition~\ref{strict_bound}, the original optimization problem for the NOMA scheme is equivalent to
\begin{maxi!}|s|[2]
{p_1}{\mu_{\text{N},1}}
{\label{NOMA_PA}}{} 
\addConstraint{\mu_{\text{N},2}}{= \beta \mu_{0}} \label{NOMA_PAb}
\addConstraint{p_1+p_2}{= p_0 },{\quad p_1,p_2 \in R^+}.\label{NOMA_PAc}
\end{maxi!}
We in the following solve the optimization problem given in~\eqref{NOMA_PA}  in the IBL and FBL regimes, respectively. 

\subsubsection{UAV Throughput Maximization under the NOMA Scheme in the IBL Regime} 
According to~\eqref{muNOMAueIBL}, constraint \eqref{NOMA_PAb} is equivalent to the following equation  

\begin{equation}
\label{NOMA_snr_gu}
\gamma^{\rm{IBL}}_{\text{gu}} \triangleq	\frac{p_2G_2}{p_1G_2+\sigma^2} = 2^{\frac{\beta \mu^{\rm{IBL}}_{0}}{M}} + 1,
\end{equation}
where $\gamma^{\rm{IBL}}_{\text{gu}}$ represents the received SINR for the GU   when decoding its own data packet.
Clearly, $\gamma^{\rm{IBL}}_{\text{gu}}$ is a constant for a given $\beta$.

According to \eqref{NOMA_PAc},  $p_2 = p_0 - p_1$ holds. Substituting this into  \eqref{NOMA_snr_gu}, the optimal solution (of the NOMA scheme in the IBL regime) to Problems \eqref{NOMA_PA_or} and \eqref{NOMA_PA} is given by

\begin{equation}\label{SINR_P_2}
	p^{\rm{IBL}}_{{{\rm{N}},1}^*}
	= p_0 - \frac{\gamma^{\rm{IBL}}_{\text{gu}} p_0G_2+\gamma^{\rm{IBL}}_{\text{gu}} \sigma^2}{G_2\left(1+\gamma^{\rm{IBL}}_{\text{gu}}\right)}.
\end{equation}

By inserting $p_1=p^{\rm{IBL}}_{{{\rm{N}},1}^*}$ and $p_2=p_0-p_1$ into \eqref{muNOMAuavIBL},  the optimal throughput of the UAV is obtained.

\subsubsection{UAV Throughput Maximization under the NOMA Scheme in the FBL Regime}
 Similarly, in the FBL regime, to guarantee  constraint   \eqref{NOMA_PAb}, the coding rate of the data packet to the GU should satisfy  $r_2 = \frac{\beta \mu^{\rm{FBL}}_0}{M(1-\nu_2)}$.   
  Hence, according to \eqref{FBLmuUserNOMA} the SINR at the GU should satisfy
 \begin{equation}
 	\gamma^{\rm{FBL}}_{\text{gu}} = \mathcal{P}^{-1}(M, \nu_2, r_2).
 \end{equation}
By substituting $p_2 = p_0-p_1$, the optimal power allocation is obtained as
\begin{equation}
    p^{\rm{FBL}}_{{{\rm{N}},1}^*}
	= p_0 - \frac{\mathcal{P}^{-1}(M, \nu_2, r_2) (p_0G_2+ \sigma^2)}{G_2\left(1+\mathcal{P}^{-1}(M, \nu_2, r_2)\right)}.
\end{equation}
Hence, we have $p_2 = p_0 - p^{\rm{FBL}}_{{{\rm{N}},1}^*}$. Then, the   optimal throughput of the UAV and the corresponding throughput of GU in the FBL regime can be immediately  calculated according to \eqref{muNOMAuavFBL} and \eqref{FBLmuUserNOMA}.

\subsection{UAV Throughput Maximization under the Relaying Scheme }
{In the relaying scheme, the total frame length of $M$ symbols is separated into two phases with lengths of $m_1$ and $m_2$.} 
Hence, to maximize the throughput of the UAV in  the IBL regime given in~\eqref{RelayIBLuav}  and in the FBL regime given in~\eqref{FBLRelay}, the optimal transmit powers $p_1$ and $p_2$ as well as  the optimal blocklengths $m_1$ and $m_2$ should be determined. The optimization problem in the relaying scheme is formulated as:
\begin{maxi!}|s|[2]
{p_1,m_1}{\mu_{\text{R},1}}
{\label{Relaying_JA}}{} 
\addConstraint{\mu_{\text{R},1}}{\geq 0} \label{Relaying_JAb}
\addConstraint{\mu_{\text{R},2}}{\geq \beta \mu_{0}  \label{guaranteeUEJL2}}
\addConstraint{m_1p_1+m_2p_2}{\leq Mp_0},{\quad p_1,p_2 \in R^+}\label{resourceP2}
\addConstraint{m_1+m_2}{\leq M \label{resourceM}},{\quad m_1,m_2 \in R^+},
\end{maxi!}

\vspace{-.3cm}
\noindent where the constraint \eqref{Relaying_JAb} indicates that the   throughput of the first phase  should be higher than the second phase. 
Similar to the discussion in Proposition~1, the inequalities in    \eqref{guaranteeUEJL2}, \eqref{resourceP2} and  \eqref{resourceM} should be satisfied with equality 
 to maximize ${\mu_{\text{R},1}}$. Hence, the above problem is equivalent to 
\begin{maxi!}|s|[2]
{p_1,m_1}{\mu_{\text{R},1}}
{\label{Relaying_JAe}}{} 
\addConstraint{\mu_{\text{R},1}}{\geq 0} \label{Relaying_JAbe}
\addConstraint{\mu_{\text{R},2}}{= \beta \mu_{0} \label{guaranteeUEJLe}}
\addConstraint{m_1p_1+m_2p_2} {= Mp_0 ,\label{resourcePe}}{\!\!\! \!\!\! \!\!\! p_1,p_2 \in R^+}
\addConstraint{m_1+m_2}{= M \label{resourceMe}},{\quad m_1,m_2 \in R^+}.
\end{maxi!}
%
In the following, we solve the problem 
in both the IBL and FBL regimes.

\subsubsection{UAV Throughput Maximization   under the Relaying Scheme in the IBL Regime}
We first propose the following key proposition to solve Problem~\eqref{Relaying_JAe} in the IBL regime. 
\begin{prop}\label{convexity_prop2}
	Problem~\eqref{Relaying_JAe} is convex in the IBL regime.  
\end{prop}
\begin{proof}
Note that  $\mu_{\text{R},2}=\beta \mu^{\rm{IBL}}_0$. Hence, we have $p_2 = \frac{\sigma^2}{G_3} (2^{\frac{ \beta \mu^{\rm{IBL}}_0 }{M-m_1}}-1)$.
To guarantee \eqref{resourcePe} and \eqref{resourceMe}, we should have
\begin{equation}
\label{p_1inproof}
  p_1 = \left(Mp_0-(M-m_1)\left(2^{\frac{ \beta \mu_0}{M-m_1}}-1\right)\sigma^2/G_3\right)/m_1.  
\end{equation}

Substituting~\eqref{p_1inproof} into  \eqref{RelayIBLuav}, the   throughput can be expressed as 
\begin{equation}\label{objectiveformulate}
\begin{split}
	\mu^{\rm{IBL}}_{\rm{R},1} =  & ~m_1\log_2\!\!\left(\!1\!+\!\frac{\left(Mp_0\!-\!(M\!-\!m_1)\frac{2^{\frac{\beta \mu_0}{M-m_1}}-1}{G_3}\sigma^2\right)G_1}{m_1\sigma^2}\right) \\
	&-\beta \mu_0,	
\end{split}
\end{equation}
in terms of the single variable $m_1$.
Hence, the proposition is proved {if $\mu^{\rm{IBL}}_{\rm{R},1}$ in \eqref{objectiveformulate} is   concave in $m_1$}. 
We show this concavity as follows.

First, we introduce an auxiliary function $t$ with respect to~$m_1$
\begin{equation}
	t  = \left(Mp_0-(M-m_1)\frac{2^{\frac{\beta \mu_0}{M-m_1}}-1}{G_3}\sigma^2\right)\frac{G_1}{\sigma^2}.
\end{equation}

We further obtain the first and second order derivatives of $t(m_1)$:  
\begin{equation}
	\frac{\partial t}{\partial m_1} = \frac{G_1}{G_3}\left(2^{\frac{\beta \mu_0}{M-m_1}}-1-\frac{\beta \mu_0 \log_e(2)}{M-m_1}2^{\frac{\beta \mu_0}{M-m_1}}\right).
\end{equation}
By introducing $g=2^{\frac{\beta \mu_0}{M-m_1}} \ge 0$, we have
\begin{equation}\label{firstpar}
	\frac{\partial t}{\partial m_1} = \frac{G_1}{G_3}\left((1-g\log_e2)2^g-1\right) \le 0.
\end{equation}
In addition, we have 
\begin{equation}\label{secpar}
	\frac{\partial^2t}{\partial m_1^2} = - \frac{G_1}{G_3}\frac{\beta^2M^2(\log_e2)^2 2^{\frac{\beta \mu_0}{M-m_1}}}{2(M-m_1)^3} \le 0.
\end{equation}

Then, the objective in (\ref{objectiveformulate}) can be   expressed as 
\begin{equation}\label{Obj4Hessian}
	\mu^{\rm{IBL}}_{\rm{R},1}  = m_1\log_2\left(1 + \frac{t }{m_1}\right) - \beta \mu_0.
\end{equation}
We can obtain the   second order derivative of $\mu^{\rm{IBL}}_{\rm{R},1}$
with respect to $m_1$ as
\begin{equation}\label{toproveconcave}
	 \frac{\partial^2 \mu^{\rm{IBL}}_{\rm{R},1} }{{\partial {m_1}^2}}= 
	 \frac{m_1^3\frac{\partial^2t}{\partial m_1^2}\!-\!m_1^2 (\frac{\partial t}{\partial m_1})^2\!-\!t ^2\!\!+\!m_1^2t \frac{\partial ^2t}{\partial m_1^2}\!\!+\!2t (\frac{\partial t}{\partial m_1})^2}{\log_e(2)(m_1+t )^2},
\end{equation}
where 
$m_1$ is by definition positive.
According to~\eqref{firstpar} and~\eqref{secpar}, $\frac{\partial^2 \mu^{\rm{IBL}}_{\rm{R},1} }{{\partial {m_1}^2}}\le0$ holds. 
\end{proof}

According to Proposition~\ref{convexity_prop2}, Problem~\eqref{Relaying_JAe} can be solved efficiently in the IBL regime.

\

\subsubsection{UAV   Throughput Maximization under the Relaying Scheme in  the FBL Regime}
Note that the $\mathcal{R}$ function in \eqref{Rfunciton} is not concave. Hence, the objective function in Problem \eqref{Relaying_JAe} is not concave in the FBL regime, i.e.,  Problem \eqref{Relaying_JAe} is not convex. In the following, we solve the problem in the FBL regime by utilizing the monotonic property.

\begin{prop}\label{Monoton_prop3}
 The objective of Problem \eqref{Relaying_JAe} is a monotonic function with respect to  two variables $p_1$ and $m_1$, if the blocklength and reliability requirements are within a range of practical interest,  i.e., $m\geq 84$ (which is the smallest blocklength in  LTE-A) and $\epsilon<10^{-1}$.
\end{prop}

\begin{proof}
Clearly the proposition holds if it is shown that $r= \mathcal{R}(m_1,\gamma_1,\epsilon_{\text{uav}})$ is monotonically increasing in both $\gamma_1$ and $m_1$. 

The first derivatives of $r$ with respect to $\gamma_1$ and $m_1$ are given by
\begin{equation}\label{partialp1}
	\frac{\partial r}{\partial \gamma_1} =  \frac{(\gamma_1+1)\sqrt{\gamma_1^2+2\gamma_1}-Q^{-1}(\epsilon_{\text{uav}})/m_1}{(1+\gamma_1)^2\sqrt{\gamma_1^2+2\gamma_1}} \log_e2 \, ,
\end{equation}
\begin{equation}\label{partialr1}
	\frac{\partial r}{\partial m_1} = \frac{1}{2} Q^{-1}(\epsilon_{\text{uav}})\sqrt{\frac{\gamma_1(\gamma_1+2)}{(\gamma_1+1)^2m_1^3}} \log_e2+\frac{1-\log m_1}{m_1^2} \, .
\end{equation}

We assume that in our scenario the SNR is higher than 0 dB, i.e., $\gamma>1$ and the practical value of $m_1>84$ and $Q^{-1}(\epsilon_{\text{uav}}) \in (2.3263, 5.9978)$. 
For \eqref{partialp1},  if $(\gamma_1+1)\sqrt{\gamma_1^2+2\gamma_1}>1$ and $Q^{-1}(\epsilon_{\text{uav}})/m_1<1$ hold, then $\frac{\partial r}{\partial \gamma_1}>0$. 
For \eqref{partialr1}, if we denote   $c=\frac{1}{2}\log_e2\cdot Q^{-1}(\epsilon_{\text{uav}})\sqrt{\frac{\gamma_1(\gamma_1+2)}{(\gamma_1+1)^2}}$ for short, the minimum of $\frac{\partial r}{\partial \gamma_1}$ is obtained at $m_1'=(\frac{2}{c})^2$. For $m_1\in (m_1',M]$, \eqref{partialr1} is increasing in $m_1$.
 For $m_1>84$, to guarantee     $\frac{\partial r}{\partial m_1}>0$, 
 $c>0.3743$ is required. As for the practical values mentioned above, $\frac{1}{2}\log_e2Q^{-1}(\epsilon_{\text{uav}})\sqrt{\frac{\gamma_1(\gamma_1+2)}{(\gamma_1+1)^2}}\geq 0.4031>0.3743$ holds. We thus have $\frac{\partial r}{\partial m_1}>0$. 
\end{proof}

As the objective function is monotonic increasing with respect to both variables $m_1$ and $p_1$, the problem can be solved by applying the framework of monotonic optimization \cite{zhang2013monotonic}, resulting in a global optimum.

\section{Feasibility Analysis on Guaranteeing the Weighted Throughput of the GU}

{So far, we have characterized the IBL and FBL throughputs under both the NOMA and relaying schemes and provided   corresponding optimal resource allocation designs while guaranteeing the weighted throughput of the GU, i.e., guaranteeing $\beta \mu_0$.} In this section, we provide a feasibility discussion on $\beta$ for both     schemes.

We start with the NOMA scheme and introduce the following proposition
\begin{prop}\label{FBLgeqIBLpro_NOMA}
    Under the NOMA scheme, the feasible range of $\beta$ values is $\beta \in [0,1]$ in both the IBL regime and FBL regime. 
\end{prop}
\begin{proof}
    Recall  that for both the IBL and FBL regimes, under the NOMA scheme the GU decodes the received signal based on the SINR $\gamma_2 = \frac{p_2G_2}{p_1G_2+\sigma^2}$. 
    Under the power constraint $p_0=p_1+p_2$, we have $\max\limits_{p_1} \gamma_2  = \frac{p_0G_2}{\sigma^2}$, which is exactly the SNR of the GU when it is being solely served by the network. In such a   case, $p_1=0$ and $p_2=p_0$, which indicates that allocating zero power to the UAV only guarantees $\beta=1$.
    Hence, the proposition is verified. 
\end{proof}

{Proposition~\ref{FBLgeqIBLpro_NOMA} indicates that when additionally serving the UAV (with a positive throughput) under the NOMA scheme, part of the GU's throughput is sacrificed.} On the other hand, the sum throughput of the UAV and the GU is boosted, as the channel from the AP to the UAV is much stronger than the one to the GU.

Next, we discuss the relaying scheme and have the following proposition
\begin{prop}\label{FBLgeqIBLpro_Relay}
    For the relaying scheme, the feasible value of $\beta$ is possible to be higher than~$1$. In such case, a win-win situation is created, i.e., the UAV is additionally served while the GU's performance is improved.  
\end{prop}

\begin{proof}
We first show that    $\beta\geq 1$ possibly holds  under the relaying scheme by considering the following example. 
We  consider
a relaying scheme
 with $m_1=m_2=M/2$.
Note that the IBL throughput to the UAV is non-negative, i.e., $\frac{M}{2}{\log _2}\left(1 + \frac{{{p_1}{G_1}}}{{{\sigma ^2}}}\right) - \beta M{\log _2}\left(1 + \frac{{{p_0}{G_2}}}{{{\sigma ^2}}}\right) \ge 0$.
Hence, we have 
\begin{equation}
 \begin{split} 
& \beta  > 1 \\
 \Leftrightarrow~~  &{\log _2}\left(1 + \frac{{{p_1}{G_1}}}{{{\sigma ^2}}}\right) - {\log _2}{\left(1 + \frac{{{p_0}{G_2}}}{{{\sigma ^2}}}\right)^2} > 0 \\
 \Leftrightarrow~~ & \left(1 + \frac{{{p_1}{G_1}}}{{{\sigma ^2}}}\right) > {\left(1 + \frac{{{p_0}{G_2}}}{{{\sigma ^2}}}\right)^2}.
\end{split} 
\end{equation}
Note that the channel from the AP to the UAV is generally   LoS and the link from the AP to the selected GU (in pair with the UAV) is generally weak, i.e., it is more likely that  ${G_1} \gg {G_2}$. Hence, there exist  feasible  $p_1$ and $p_0$, which satisfy    $\left (1 + \frac{{{p_1}{G_1}}}{{{\sigma ^2}}}\right) > {\left(1 + \frac{{{p_0}{G_2}}}{{{\sigma ^2}}}\right)^2}$.

The possibility of  $\beta\geq 1$ can be also proved in the IBL regime   
based on the following Lemma~\ref{FBLgeqIBLpro_Relay2} which provides that  if $\beta>1$ is   guaranteed in the IBL regime, it also possibly holds in the FBL regime.  
\end{proof}

\begin{lem}\label{FBLgeqIBLpro_Relay2}
    In the relaying scheme with   given $\beta$,   $m_1$ and $m_2$, under certain conditions guaranteeing GU's transmission quality in the FBL regime requires less transmit power  $p_2$ to GU  than guaranteeing it in the IBL regime. 
\end{lem}

\begin{proof}
We prove Lemma~\ref{FBLgeqIBLpro_Relay2} by discussing the special case in which we have     $\beta = 1$, $m_1=m_2=M/2=m$ and the reliability constraint of UAV is much lower than the GU, i.e., $\nu_1 \ll \nu_2$. In addition, we consider a reliable transmission scenario where the SNRs of all links are higher than 1. 


In the IBL regime, according to \eqref{RalayIBLue},  the throughput requirement of the GU is  $\log_2(t_2) = 2\log_2{t_0}$,  where we have defined
 $t_0=1+\frac{p_0G_2}{\sigma^2}$. 
Then, the IBL throughput of the UAV is obtained by
\begin{equation}
   	c = \log_2(t_1)-\log_2(t_2) =\log_2\frac{t_1}{t_2},  
\end{equation}
 where $t_1=1+\frac{p_1G_1}{\sigma^2}$ and $t_2=1+\frac{p_2G_3}{\sigma^2}$.

In the FBL regime, the GU's throughput requirement is given in \eqref{FBL_GU_requirement}. 
Noting that $\nu_1 \ll \nu_2$,  we have $\epsilon_{\rm{R},2}=\nu_2- \nu_1 + \nu_1\epsilon_{\rm{R},2} \approx \nu_2$, i.e., 
$Q^{-1}(\nu_2) \approx Q^{-1}(\epsilon_{\rm{R},2}) $.
We define   
$A =  Q^{-1}(\nu_2) \log_2 e \approx Q^{-1}(\epsilon_{\rm{R},2}) \log_2 e$   and $B = \frac{\log_2 m_1}{m_1} = \frac{\log_2 m_2}{m_2}$.
In addition, we denote $t_2'=1+\frac{p_2'G_3}{\sigma^2}$ and $t_1'=1+\frac{p_1'G_1}{\sigma^2}$ where $p_1'$ and $p_2'$ are the power allocation results in the FBL regime. Then, the GU's throughput requirement   given in \eqref{FBL_GU_requirement} can be represented by
\begin{equation}
 	\log_2t_2'-A\sqrt{1-\frac{1}{(t_2')^2}} - B = 2 \left(\log_2t_0-A\sqrt{1-\frac{1}{t_0^2}} - B\right).
\end{equation}
Hence, we have
\begin{equation}
 	\log_2t_2' = 2 \log_2t_0 - A\left(2\sqrt{1-\frac{1}{t_0^2}}-\sqrt{1-\frac{1}{t_2'^2}}\right) - B.
\end{equation}

Noting that SNRs are higher than 1, we have both $t_0>2$ and $t_2'>2$ hold. Therefore $2\sqrt{1-\frac{1}{t_0^2}}-\sqrt{1-\frac{1}{t_2'^2}}\geq 0$ by transforming into $3\geq \frac{4}{t_0^2} - \frac{1}{t_2'^2}$, which is obviously valid. Thus $\log_2t_2' < \log_2t_2$ holds, which is equivalent to $p_2'<p_2$. This indicates that more resources could be allocated for the first hop, i,e, the serving quality of UAV is better than that calculated in the IBL regime. 
\end{proof}

\section{Simulation Results and Discussion}

In this section, we provide our simulation results. 
We first validate our analytical model and subsequently evaluate the IBL and FBL performances of the considered network. The   simulation results are conducted under the following parameter setups. First, we consider a 3-dimensional topology where the AP is deployed at point $(0,0,20)$, while the UAV and GU positions are set to $(100,0,100)$ and $(700,0,0)$. 
In addition, the total transmit power  is   $p_0 = 1 W =30 \rm{dBm}$ while  the noise power 
is $\sigma^2 = -80 \rm{dBm}$. Moreover, we set the frame length to $M=400$ symbols. Finally, for the analysis in the FBL regime, the target error probabilities of the transmissions to the UAV and the GU are set to $\nu_1=10^{-4}$ and $\nu_2=10^{-3}$, respectively. 
Moreover, noting that the    AP-GU link, the AP-UAV link, and the UAV-GU link have different    probabilities of LoS, in the simulation we consider the following path loss formulations. 
We adopt the general path loss model from~\cite{rappaport1996wireless}, which  is given by  
\begin{equation}
	L_i = L(d_0) + 10\alpha_i \cdot \log\left(\frac{d_i}{d_0}\right) \, , 
\end{equation}
where $L(d_0)$ is determined by the free space path loss at reference distance $d_0 = 100$m. In addition,  $\{\alpha_i\}$ for $i = 1,2,3$  are the path loss exponents, which are potentially different for the three links with different probabilities of LoS. 
In particular, we adopted the setups from~\cite{rappaport1996wireless} and set  $\alpha_1 = 2$ for the AP-UAV link (with LoS) and $\alpha_2=3.5$ for the AP-GU link (without LoS). 
Moreover,  following the model in~\cite{azari2018ultra}, we have the following path loss exponent for the UAV-GU link:
	$\alpha_3 = a_1\cdot P(\theta) + b_1$,
where $P(\theta) = \frac{1}{1+a_2e^{-b_2(\theta-a_2)}}$ is the probability of the LoS between UAV and GU, as introduced in \cite{al2014optimal}.  
In addition, 
we set $a_1 = -1.5$, $a_2 = 2$, $a_2=9.61$, and $b_2=0.16$, according to the rural area setup in~\cite{alzenad20183}.

\subsection{Validations}
In this subsection, we validate our analytical model via simulation and numerical results that depict the instantaneous performances within a single frame. 
\begin{figure}[t]
    \centering
    \subfigure[The IBL throughput of UAV  ($\beta=0.5$).]{
    \label{fig:validprop1_beta4}
    \includegraphics[width=0.95\linewidth, trim= 0 0 0 0]{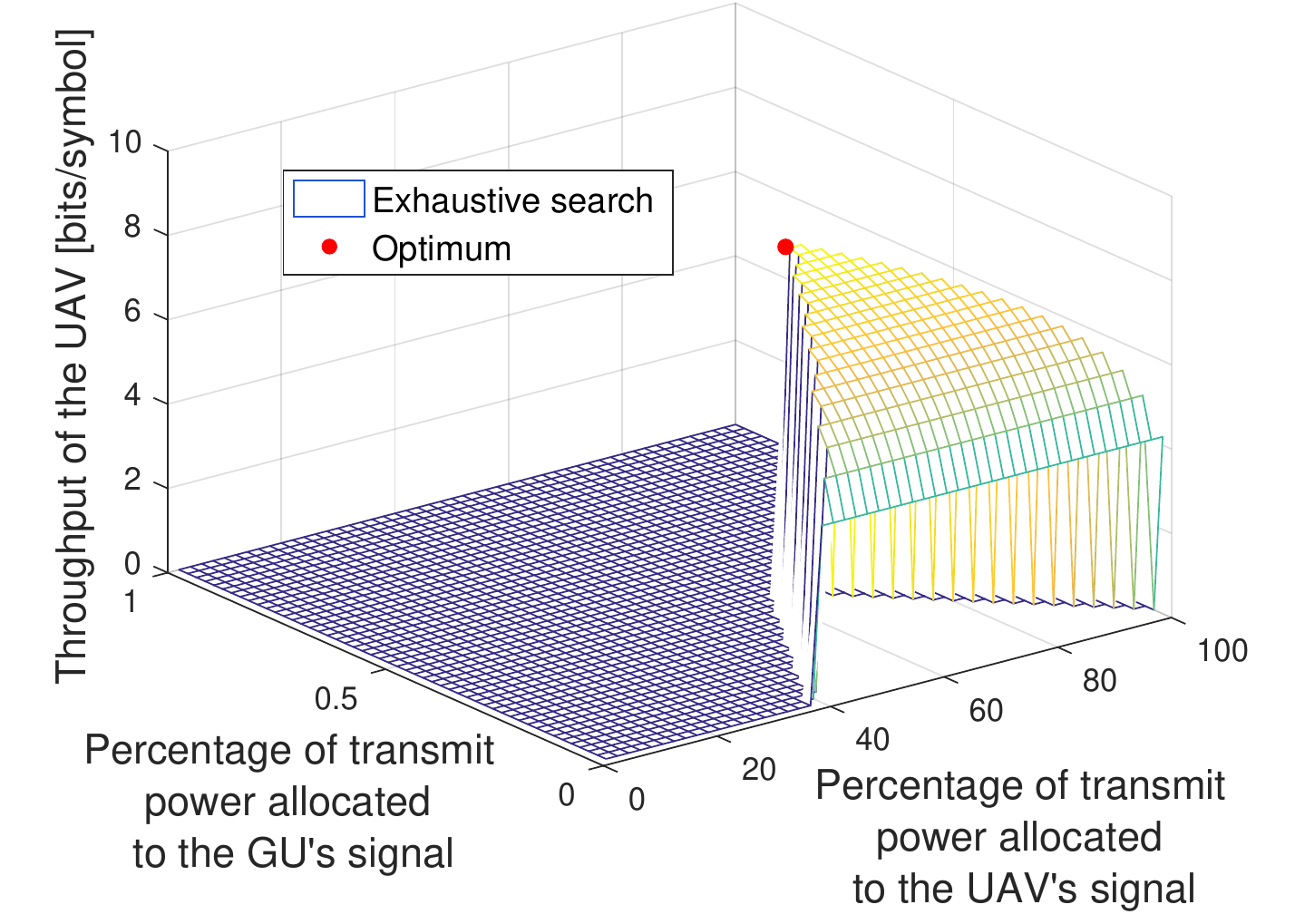}
    }
    \subfigure[Power allocation with different setups of $\beta$]{
    \label{fig:validprop1_betaMul}
    \includegraphics[width=0.95\linewidth, trim= 0 0 0 0]{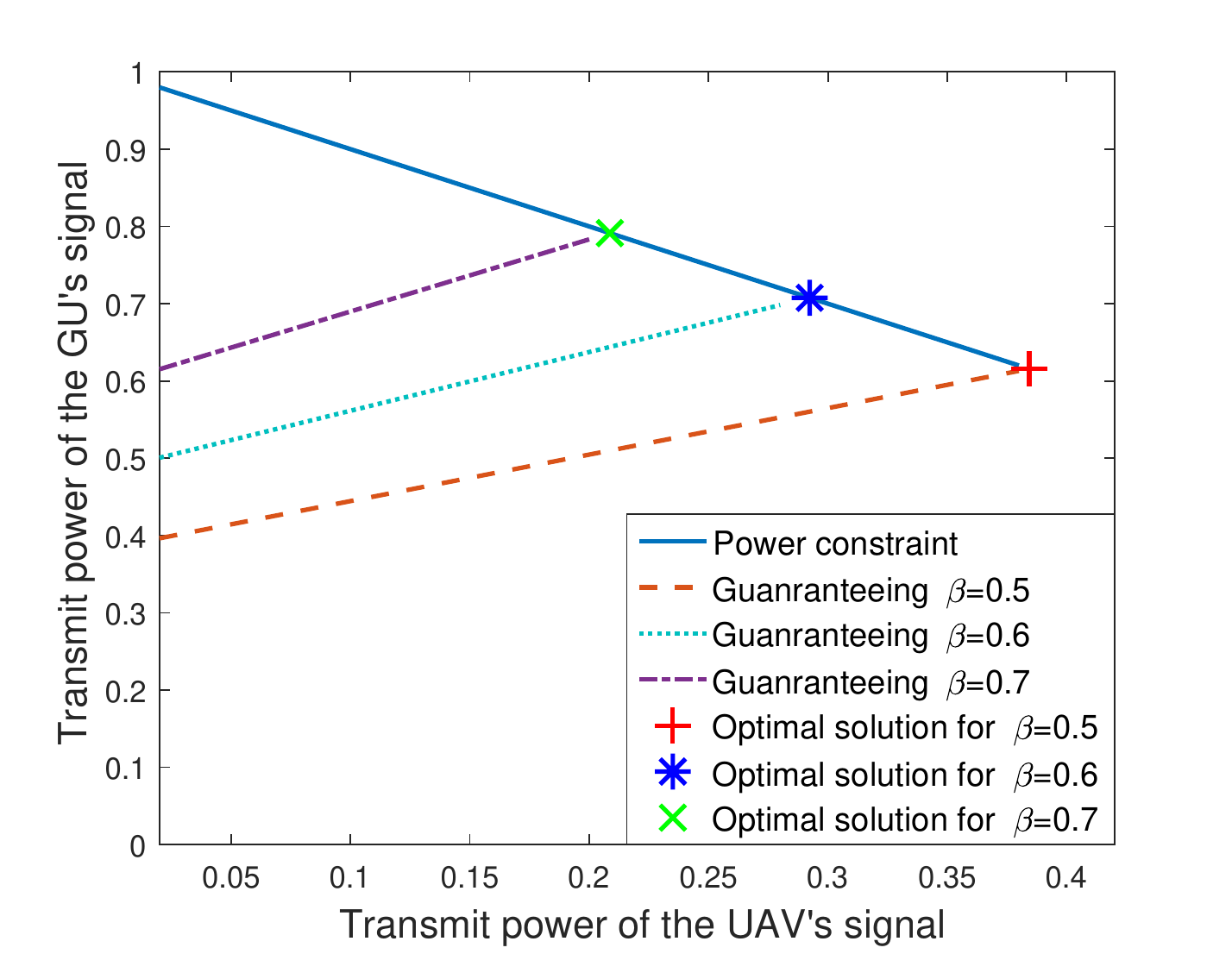}
    }
    \caption{Validation of optimal method by comparing with exhaustive search.}
    \label{fig:validprop1}
    \vspace{-.3cm}
\end{figure}
\begin{figure}[t]
    \centering
    \includegraphics[width=0.95\linewidth, trim = 0 0 0 0]{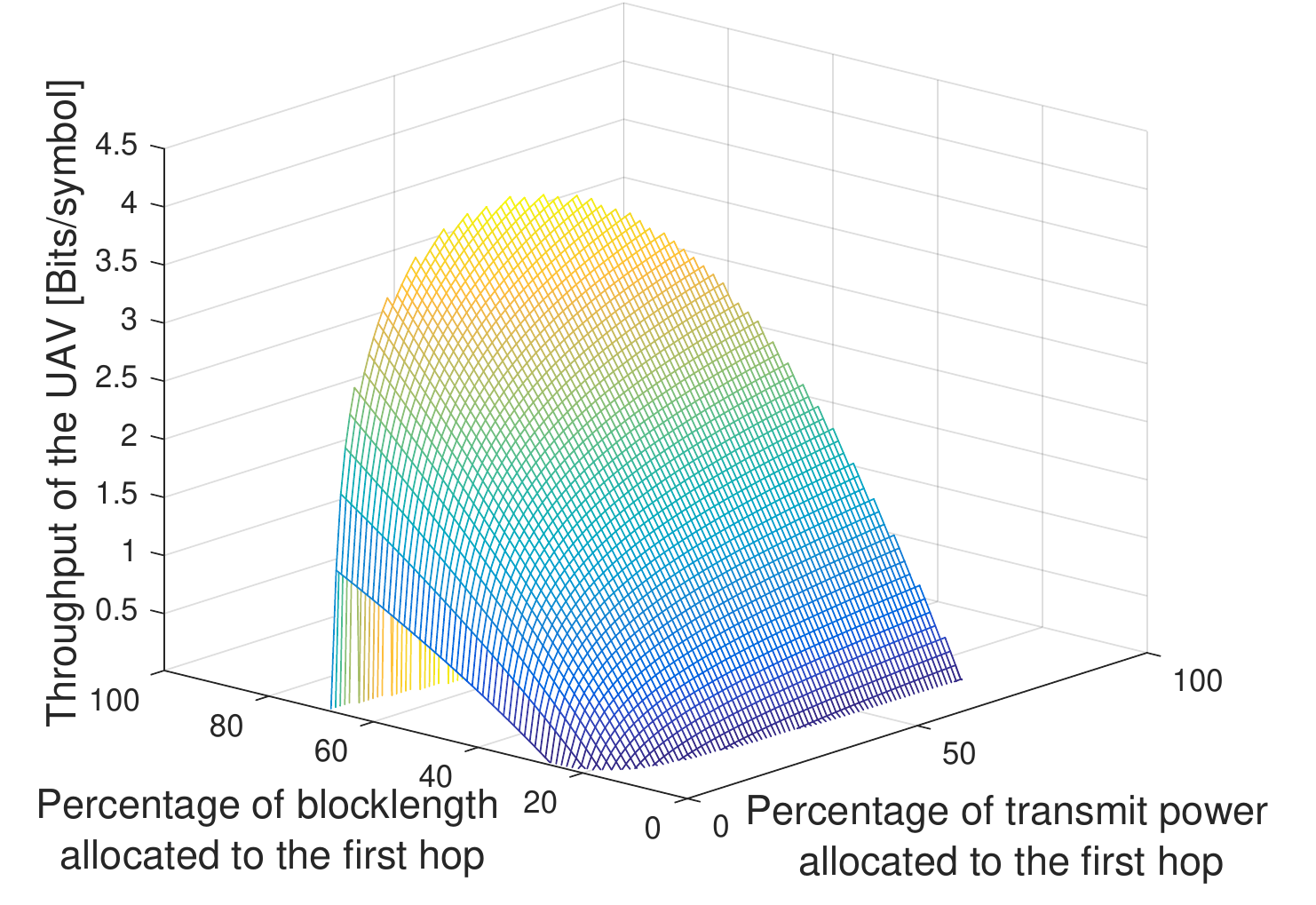}
    \caption{Validation of convexity of the relaying scheme in the IBL regime.}
    \label{fig:validprop2}
\end{figure}

We start with Fig.~\ref{fig:validprop1} to show the impact of power allocation on the UAV throughput.
{ The figure has two parts, i.e.,  Fig.~\ref{fig:validprop1_beta4} and Fig.~\ref{fig:validprop1_betaMul}. } 
 In particular,  Fig.~\ref{fig:validprop1_beta4} provides the all the feasible values of   IBL throughput of the UAV under the NOMA scheme, i.e.,  the objective in Problem \eqref{NOMA_PA_or},
while satisfying the inequality  constraints \eqref{NOMAg} and  \eqref{NOMAp}. 
In addition, the optimal solution to Problem \eqref{NOMA_PA_or} obtained via exhaustive search is also provided in the figure as a red point. 
The two cut-surfaces crossing the  optimal solution point  actually represent the equality cases of \eqref{NOMAg} and  \eqref{NOMAp}. 
Clearly, Fig.~\ref{fig:validprop1_beta4} shows that under the scenario that $\beta=0.5$, the optimal solution satisfies       \eqref{NOMAg} and  \eqref{NOMAp} with equality. Moreover, a set of corresponding results are provided in Fig.~\ref{fig:validprop1_betaMul}, where different values of $\beta$ are considered. Again, it is observed that the optimal solutions are exactly the crossing points where the  power constraint and the GU's throughput constraint hold with equality, verifying Proposition~1.

Next, we consider the relaying scheme in the IBL regime. 
In Fig.~\ref{fig:validprop2}, we show the relationship among the IBL throughput of UAV,  blocklength $m_1$ and transmit power $p_1$,   when $\beta = 1.0$. 
Note that only the feasible vaules of the UAV's throughput are visible, i.e. the blocklength and power allocation must satisfy the constraints in \eqref{Relaying_JA}. 
Clearly, it is observed that the UAV throughput in the IBL regime is concave in $m_1$ and $P_1$, confirming the characterization in Proposition~\ref{convexity_prop2}.

In addition, in Fig.~\ref{fig:validprop3_BL} we investigate the relationship between the throughput of the UAV with the blocklength $m_1$, while different values of the transmit power $p_1$ are considered.
{ Fig.~\ref{fig:validprop3_BL} shows that the UAV throughput is monotonically increasing in the blocklength of the first hop within the feasible set. Note that the sudden cut after each curve achieving its maximum value is due to the feasibility criterion, i.e.  the  constraint to guarantee GU's performance is violated   if a blocklength longer than that required at the maximum point is allocated.}
In parallel, the impact of transmit power $p_1$ on the throughput of the UAV is given in Fig.~\ref{fig:validprop3_P} for different values of the blocklength $m_1$. 
{It is observed that the UAV throughput is monotonically increasing in the transmit power allocated to the first hop within the feasible set.
In other words, 
Fig.~\ref{fig:validprop3} confirms    Proposition~\ref{Monoton_prop3} that the UAV throughput under the relaying schemes is a monotonic function with respect to 
 $p_1$ and $m_1$ within the corresponding feasible sets.  }

\begin{figure}[t]
    \centering
    \subfigure[Monotonic increasing throughput in blocklength]{
    \label{fig:validprop3_BL}
    \includegraphics[width=0.95\linewidth, trim= 0 0 0 0]{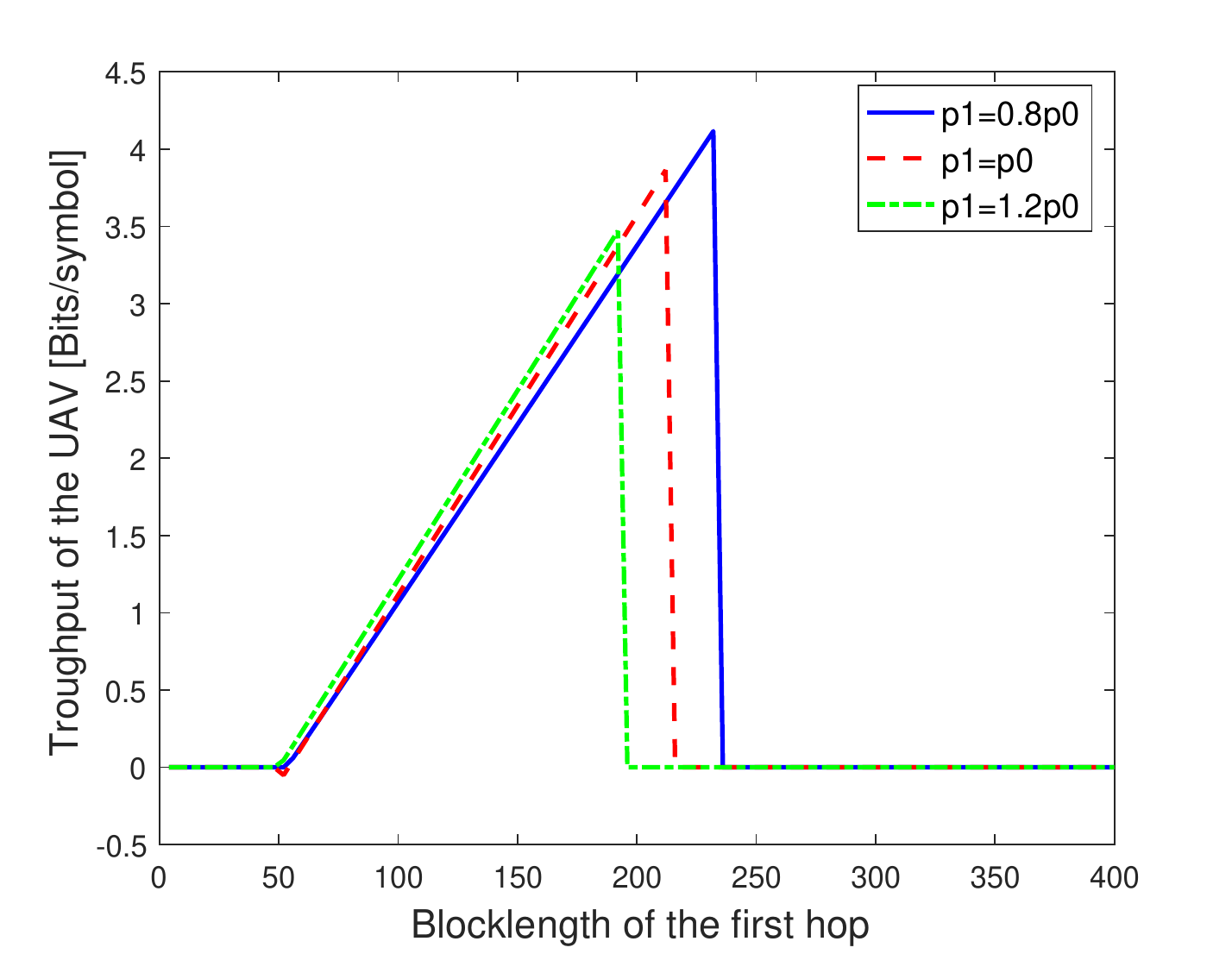}
    }
    \subfigure[Monotonic increasing throughput in power]{
    \label{fig:validprop3_P}
    \includegraphics[width=0.95\linewidth, trim= 0 5 0 -10]{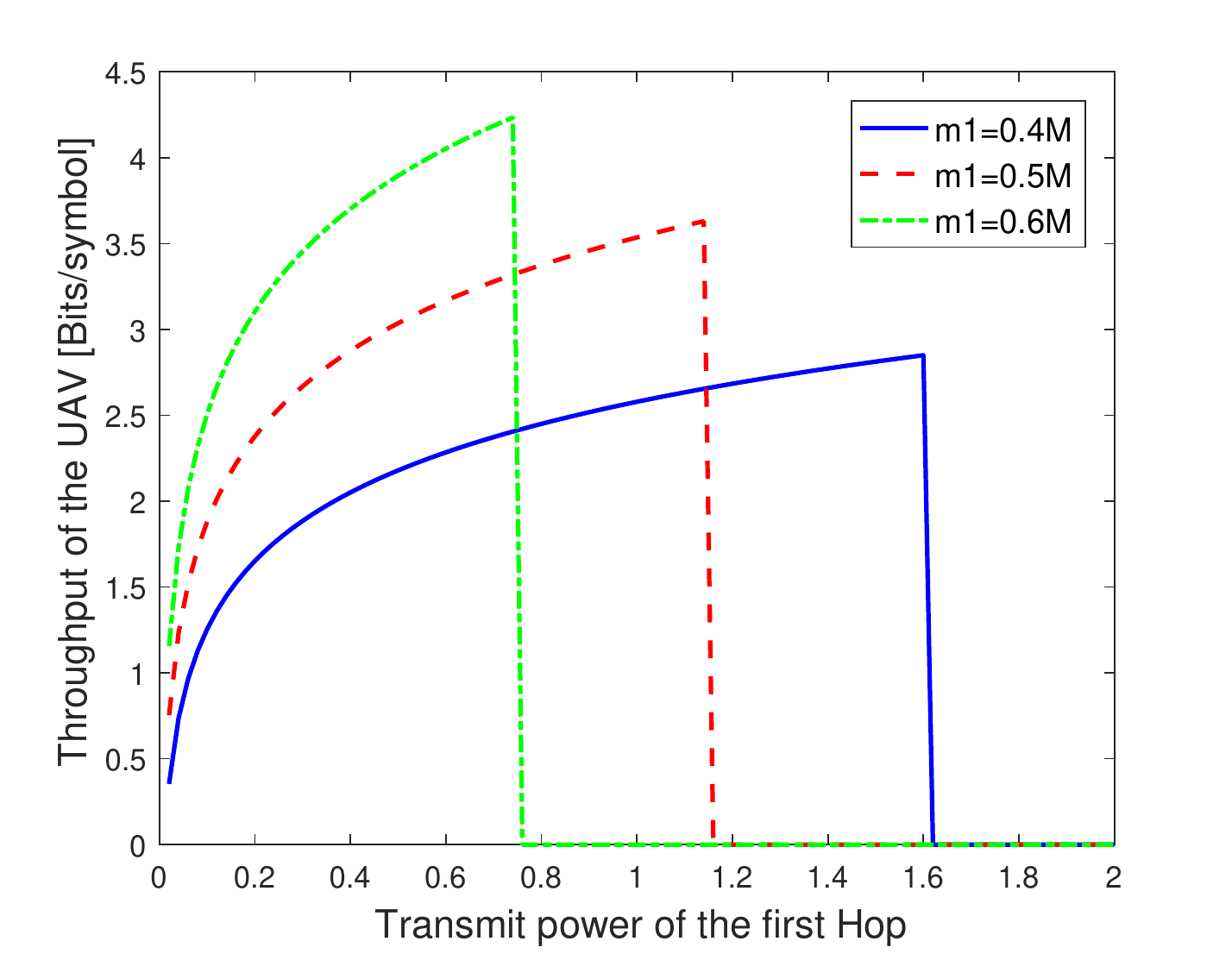}
    }
    \caption{Validation of monotonicity of relaying scheme in the FBL regime.}
    \label{fig:validprop3}
      \vspace{-.27cm}
\end{figure}
\begin{figure}[t]
	\centering
	\includegraphics[width = 0.84\linewidth, trim= 10 10 0 -23]{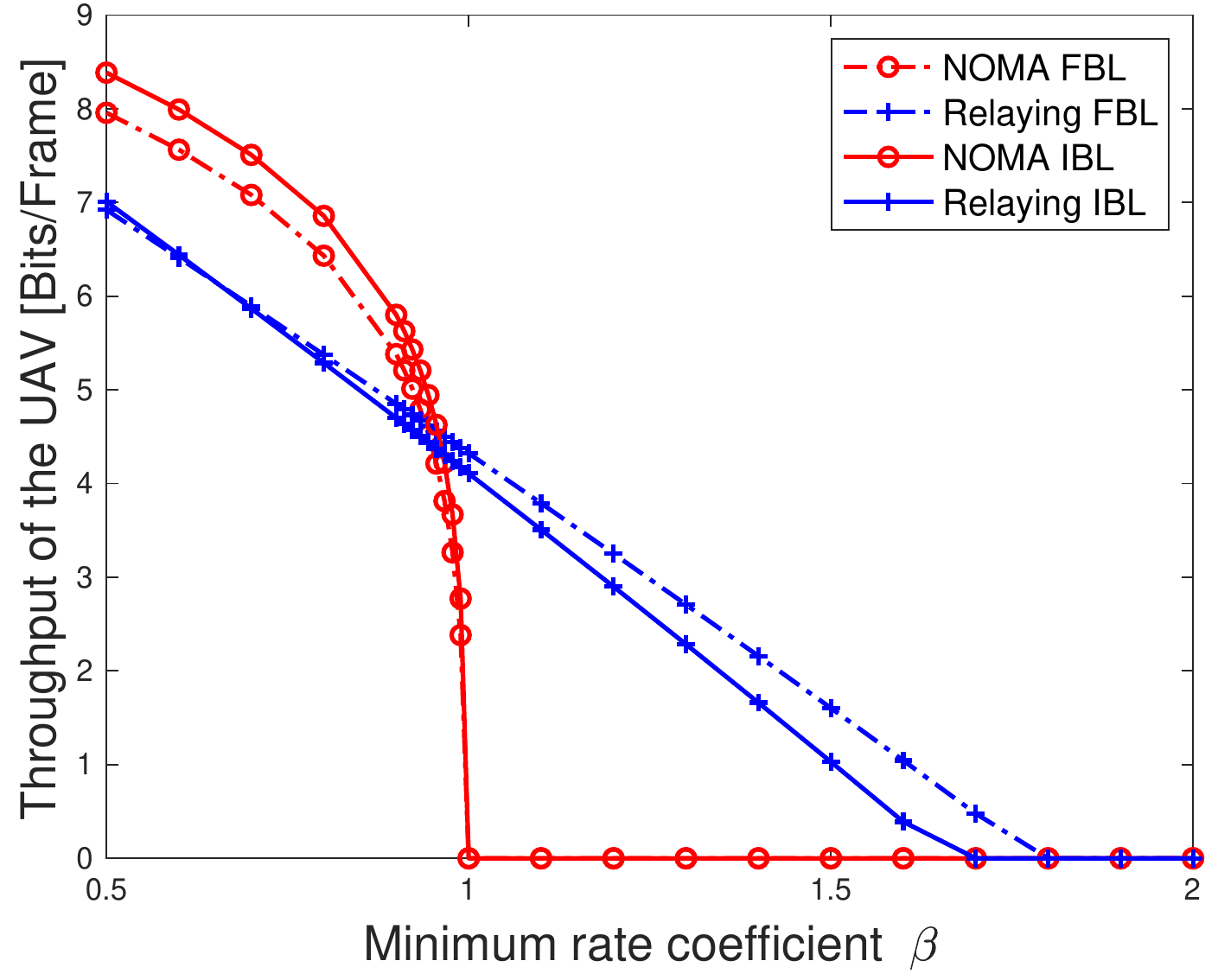}
	\caption{Performance comparison between NOMA and Relaying in both IBL and FBL regimes for GU at $700$ m.}
	\label{fig:EdgeUser}
\end{figure}

Moreover, Fig.~\ref{fig:EdgeUser} studies the throughput of the UAV while comparing the NOMA and the relaying schemes in both the IBL and the FBL regimes. In the simulation, we vary the minimum rate coefficient     $\beta$ from 0.5 to 2. The throughput of the UAV under the NOMA scheme is depicted using circle markers, while the performance under the relaying scheme is shown with a plus marker. The solid lines represent the performance in the IBL regime and the dashed line is indicating the performance in the FBL regime. 

It is evident in Fig. \ref{fig:EdgeUser} that the feasible ranges of the  minimum rate coefficient  are different in the two schemes.
{In particular, the NOMA curves drop rapidly as $\beta$ gets close to $1$.} This confirms   Proposition~\ref{FBLgeqIBLpro_NOMA} that  the NOMA scheme only serves the UAV with $\beta \in [0,1]$ in both the IBL and FBL regimes. 
On the other hand, the relaying scheme could offer more throughput ($\beta > 1$) to the GU and serve the UAV simultaneously, i.e., offer a win-win strategy, which matches our characterization in Proposition~\ref{FBLgeqIBLpro_Relay}.  
In addition, we could observe from Fig. \ref{fig:EdgeUser} that under the NOMA scheme, the throughput of the UAV in the FBL regime is smaller than the throughput in the IBL regime. 
More interestingly, the FBL throughput under the relaying scheme is higher than the IBL throughput when $\beta$ is higher than 0.6. This   verifies Proposition~\ref{FBLgeqIBLpro_Relay2} that under the relaying scheme,  it is possible that less resource is required in the FBL regime (than in the IBL regime) to guarantee the GU's transmission quality, i.e., more resources could be spent for the UAV.   
\begin{figure}[t]
	\centering
	\includegraphics[width = 0.95\linewidth, trim= 5 7 0 21]{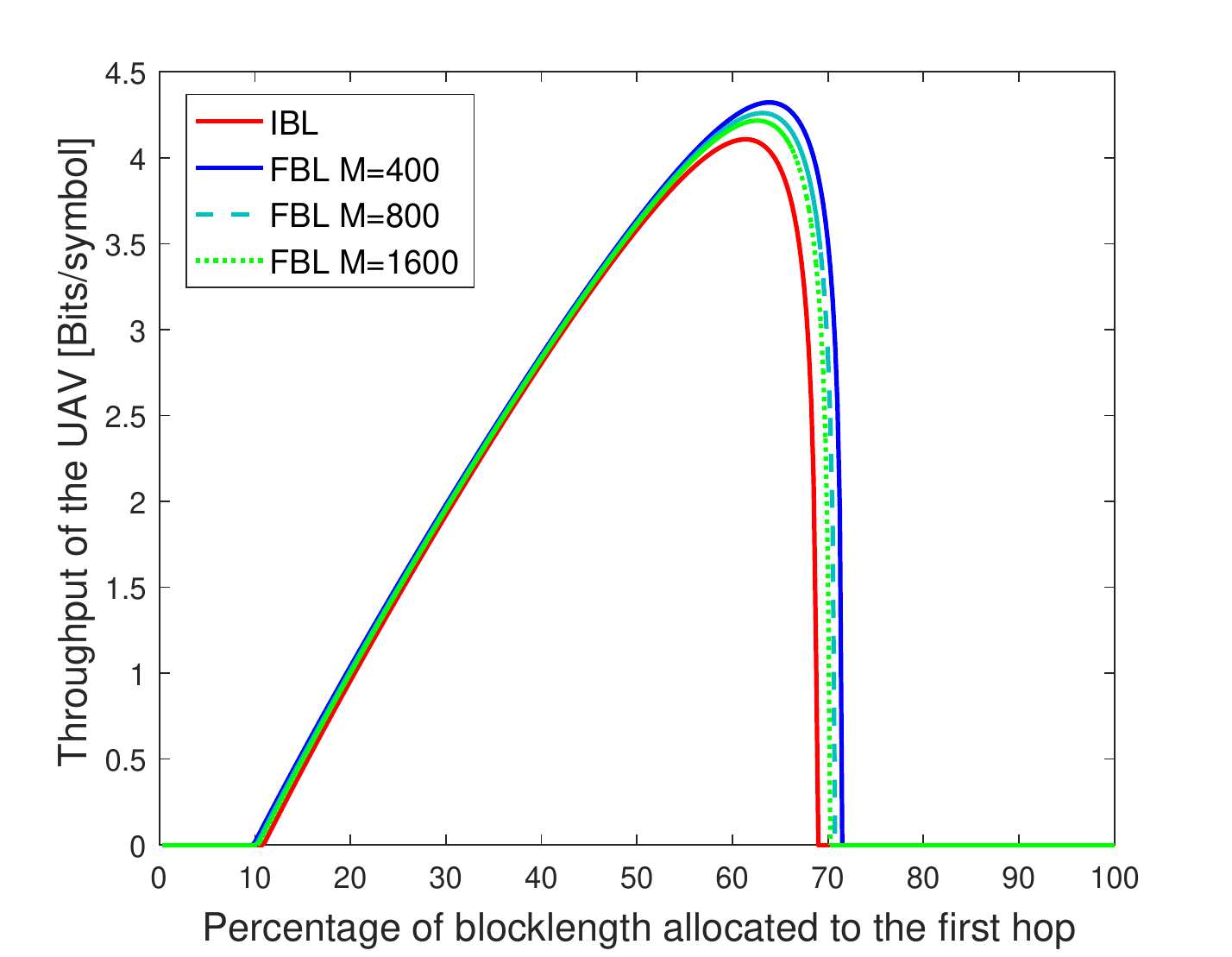}
	\caption{The UAV throughput in the IBL and the FBL regimes. In the figure, we set $\beta=1$.}
	\label{FBLgeqIBL}
\end{figure}

Next, we      show the impact of  $m_1$ on the throughput of the UAV  in Fig.~\ref{FBLgeqIBL} where     results in both the IBL and FBL regimes are provided.  
The difference between  Fig.~\ref{FBLgeqIBL}  and Fig.~\ref{fig:validprop3_BL} is that the transmit power $p_1$ in Fig.~\ref{FBLgeqIBL} is optimally chosen.
It is observed that the optimal blocklength of the first hop in the FBL regime   is shorter than the one in the IBL regime, which again confirms Lemma~\ref{FBLgeqIBLpro_Relay2}.
Moreover, this blocklength difference between the two regimes is getting smaller as the total frame length increases. 

To discuss the applicability of the model under the proposed schemes, we place another GU at $(400,0,0)$, which is closer to the AP than the setup in Fig.~\ref{fig:EdgeUser}, where the single GU is located at $(700,0,0)$. The throughput results are illustrated in Fig. \ref{MiddleRange}. 
\begin{figure}[t]
	\centering
	\includegraphics[width = 0.88\linewidth,trim= 0 15 0 -5]{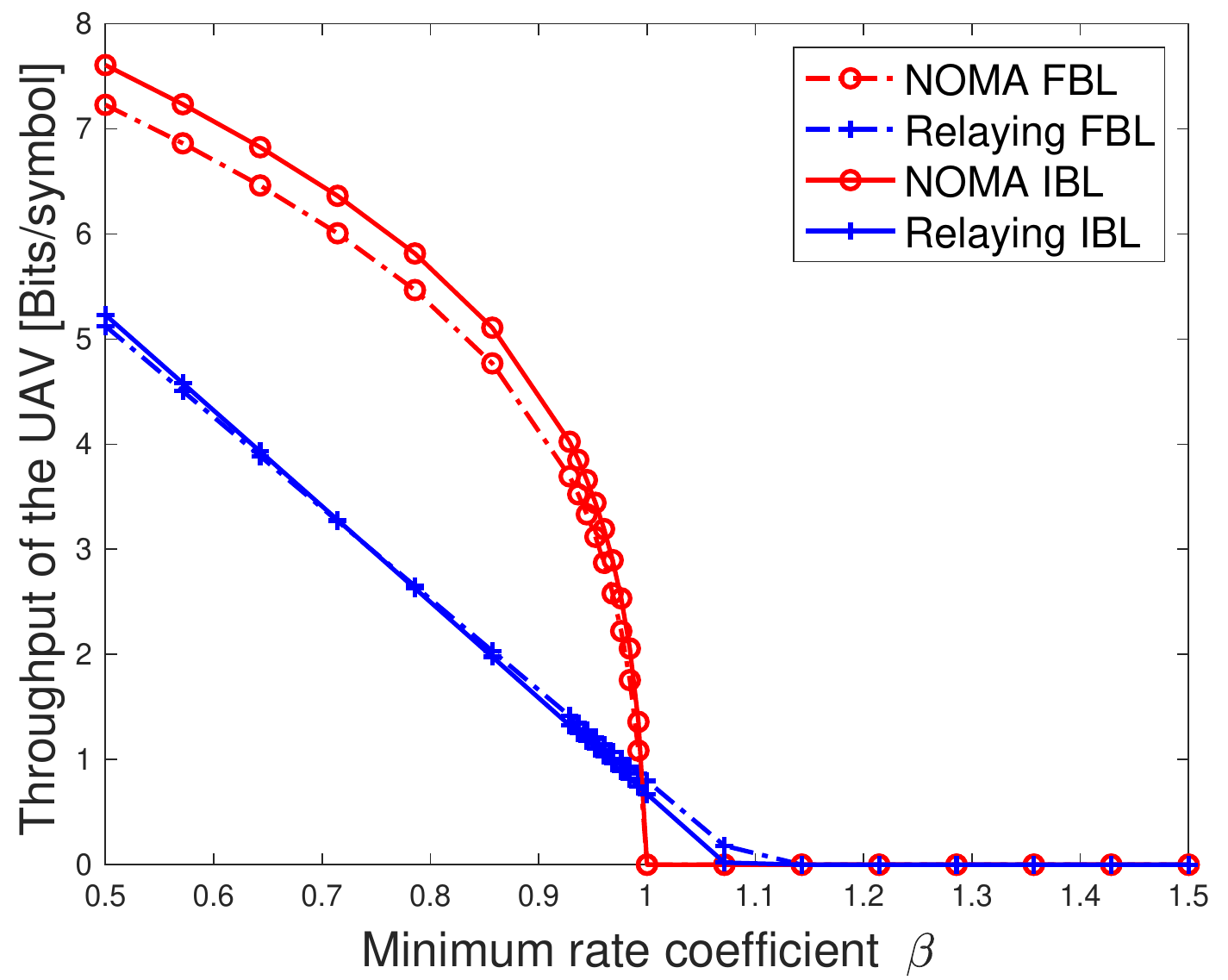}
	\caption{Performance comparison between NOMA and Relaying in both IBL and FBL regime for GU at $400$ m.}
	\label{MiddleRange}
\end{figure}
In contrast to Fig.~\ref{fig:EdgeUser}, the win-win region ($\beta>1$ and UAV throughput is positive) of relaying  is much smaller in Fig.~\ref{MiddleRange}. 
On the other hand, for $\beta<0.99$, the NOMA scheme outperforms the relaying scheme in both the IBL and FBL regimes. In other words, when the UAV is paired with the GU which has a good channel link to the AP, the NOMA scheme is generally more preferred in terms of the total throughput.

\subsection{Evaluation}
In this subsection, we evaluate the average system performance over channel fading. 
In particular, we have the following setup on the channel fading: (A) for the AP-UAV channel with LoS, we apply Rician fading with non-centrality parameter as $1$;
(B) for the UAV-GU channel, we apply Rician fading with non-centrality parameter as the probability of LoS;
(C) finally we apply Rayleigh fading model for   the AP-GU channel without LoS.

\subsubsection{User Selection}
Fig.~\ref{Distance_evaluation} illustrates the impact of the GU's distance (to the AP) on the UAV throughput performance. Note that the NOMA scheme never works at $\beta>1$. Hence,  we set the  minimum rate coefficient $\beta = 0.95$, and evaluate the performance of the NOMA and relaying schemes    in both the IBL and FBL regimes. 
\begin{figure}[t]
	\centering
	\includegraphics[width = 0.95\linewidth,trim= 0 15 0 10]{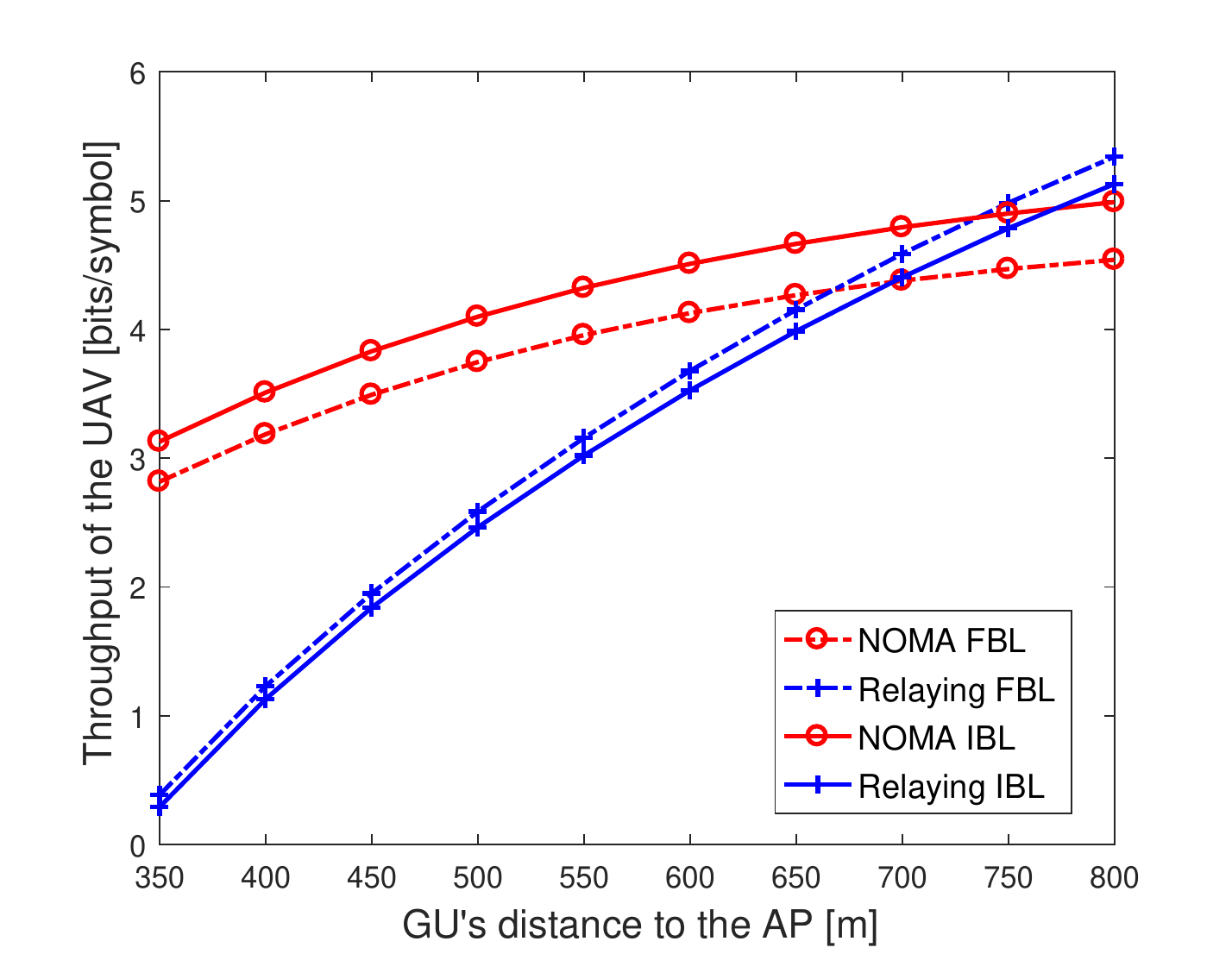}
	\caption{Performance evaluation in different GU's distance}
	\label{Distance_evaluation}
\end{figure}
The results are consistent with  Fig.~\ref{fig:EdgeUser} and Fig.~\ref{MiddleRange} in that the NOMA scheme is more preferred when the GU is close to the AP. 
Moreover, Fig.~\ref{Distance_evaluation} also confirms   the results observed in   Fig.~\ref{fig:EdgeUser} to Fig.~\ref{MiddleRange}    that in the FBL regime relaying provides a relatively higher throughput to UAV than in the IBL regime.

\subsubsection{The Impact of FBL Regime in Resource  Allocation}

In the following, we investigate how the optimization problem in the FBL regime differs from that in the IBL regime, and  what  the performance loss is if  we directly apply the allocation parameters obtained in the IBL regime into practice. 
It should be pointed out that the optimal allocation parameter of the NOMA scheme in the IBL regime is not feasible in the FBL regime. It is due to the fact that the power $p_2$ obtained in the IBL regime is not able to  guarantee the GU's reliability and throughput requirement in the FBL regime.
Therefore, we only 
  compare the IBL throughput with the FBL throughput under the relaying scheme. 
In particular, three types of results are provided: (A) IBL throughput: the optimal throughput of UAV purely according to the design  in the IBL mode; (B) FBL throughput: the optimal throughput of UAV based on the  proposed design in the FBL regime; (C) FBL throughput with IBL blocklength: we obtain the optimal blocklength solution in the   IBL regime (i.e., using the IBL throughput metrics), and calculate the FBL performance based on these blocklength (BL) results (where  the transmit power is reallocated to satisfy the GU's reliability and throughput requirements).

\begin{figure}[b]
	\centering
	\includegraphics[width = 0.95\linewidth,trim= 0 10 0 25]{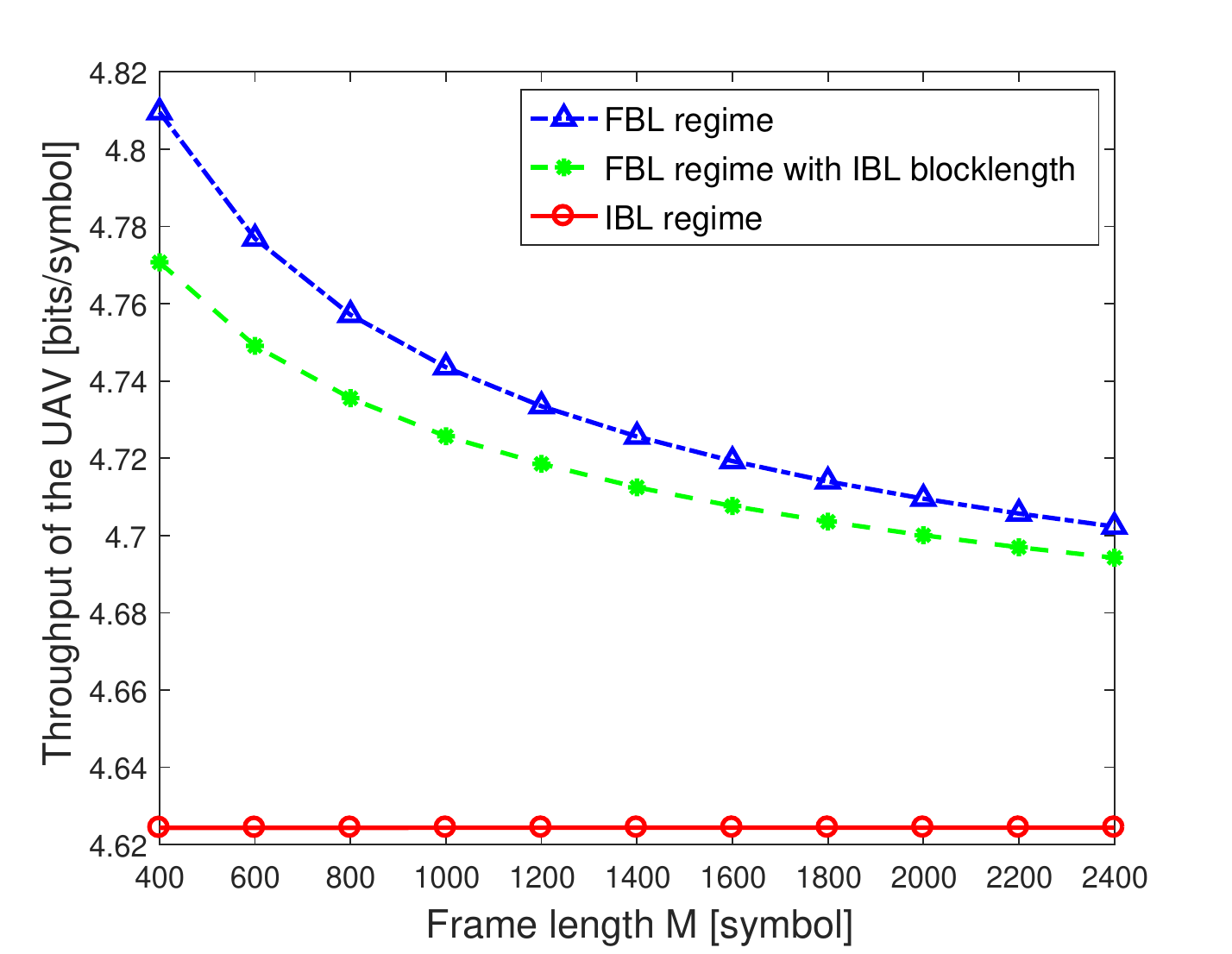}
	\caption{Throughput comparison under the relaying scheme with a varying frame length.}
	\label{BLAnalysis}
\end{figure}
The impact of the frame length $M$ on the optimal throughput of the UAV 
is shown in Fig. \ref{BLAnalysis}. 
Significant performance difference between the designs in the two regimes are observed clearly from the figure. On the one hand, the FBL throughput of UAV is much higher than the IBL throughput, i.e., the IBL model is inaccurate when the frame length $M$ is relatively short.   On the other hand, directly applying the IBL optimal solution to a system operating with FBL codes also introduces considerable performance loss, i.e., the gap between the FBL throughput and FBL throughput with IBL BL (i.e., using the BL results obtained by considering the IBL throughput characterizations).

Moreover, we study the impact of the GU's reliability requirement on the UAV throughput in Fig.~\ref{ErrAnalysis}. 
\begin{figure}[t]
	\centering
	\includegraphics[width = 0.95\linewidth, trim= 0 10 0 10]{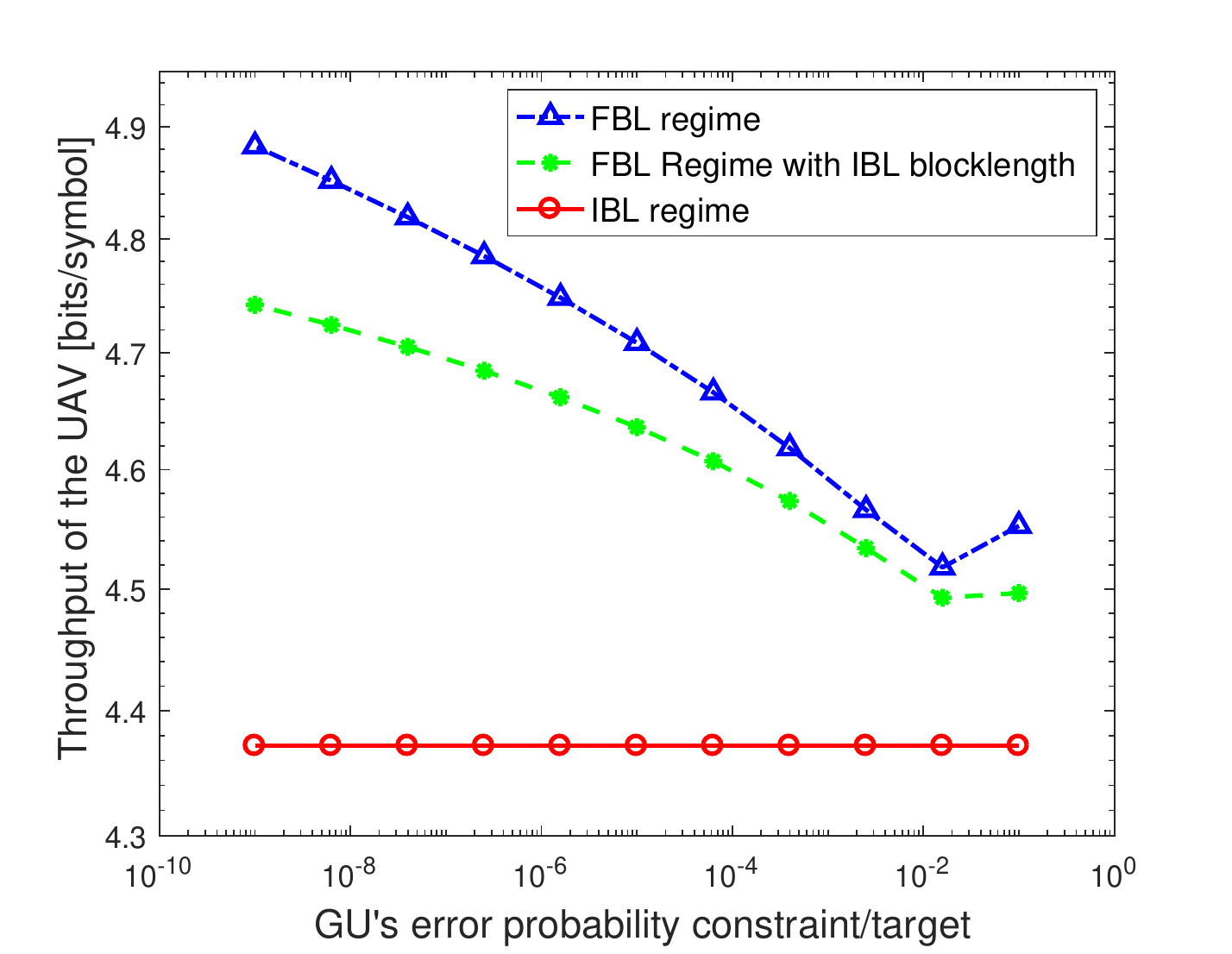}
	\caption{Throughput comparison under the relaying scheme,  where we vary the GU's error probability constraint from $10^{-9}$ to $10^{-1}$.}
	\label{ErrAnalysis}
\end{figure}
Again, significant performance gaps are observed between the FBL and IBL regimes, which confirms that it is necessary and essential to   provide the FBL  design (under  relaying scheme)    for a UAV network operating with short blocklengths. 
Moreover, this necessity is more obvious when the GU has a stringent reliability constraint, i.e.,  the performance gaps between  the curves are more considerable when the  GU's error probability constraint becomes lower. 
\begin{figure}[b]
	\centering
	\includegraphics[width = 0.82\linewidth,trim= -5 10 12 4]{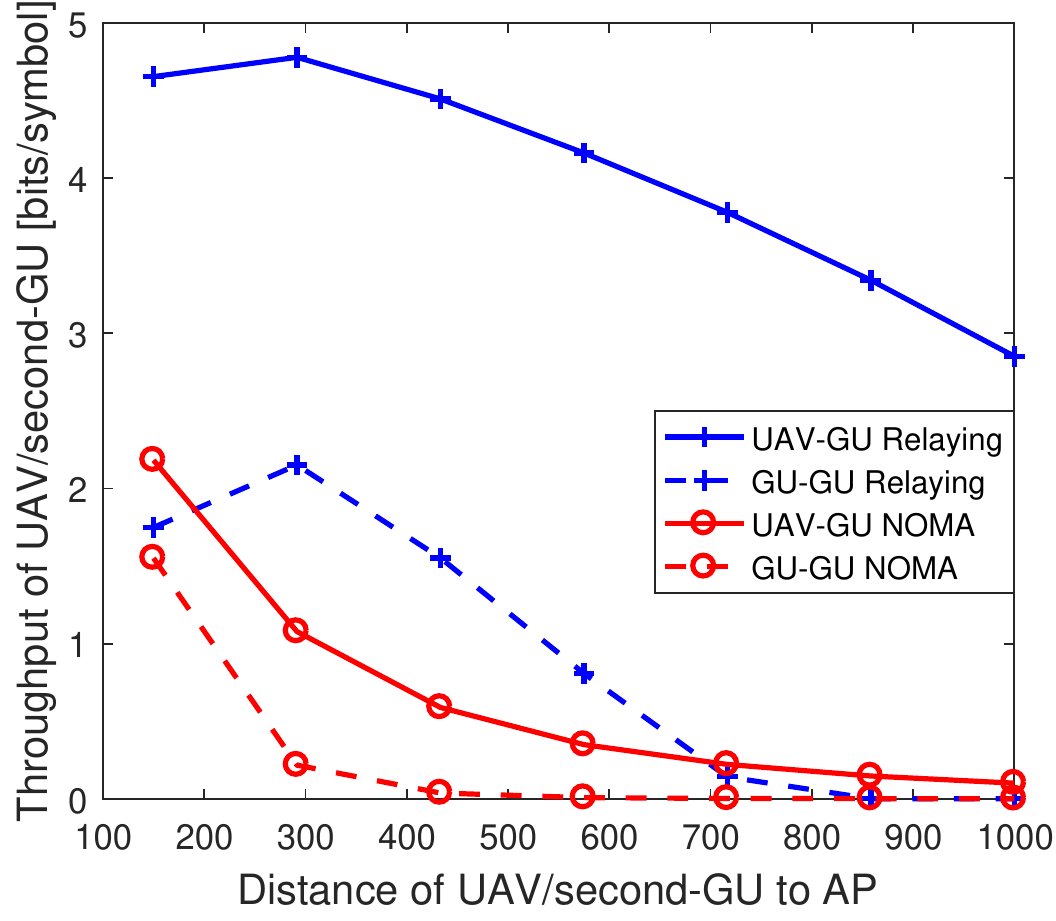}
	\caption{Throughput comparison between the UAV-GU pair and the GU-GU pair.}
	\label{throughput_comp}
\end{figure}
\begin{figure}[b]
	\centering
	\includegraphics[width = 0.90\linewidth, trim= 5 12 0 0]{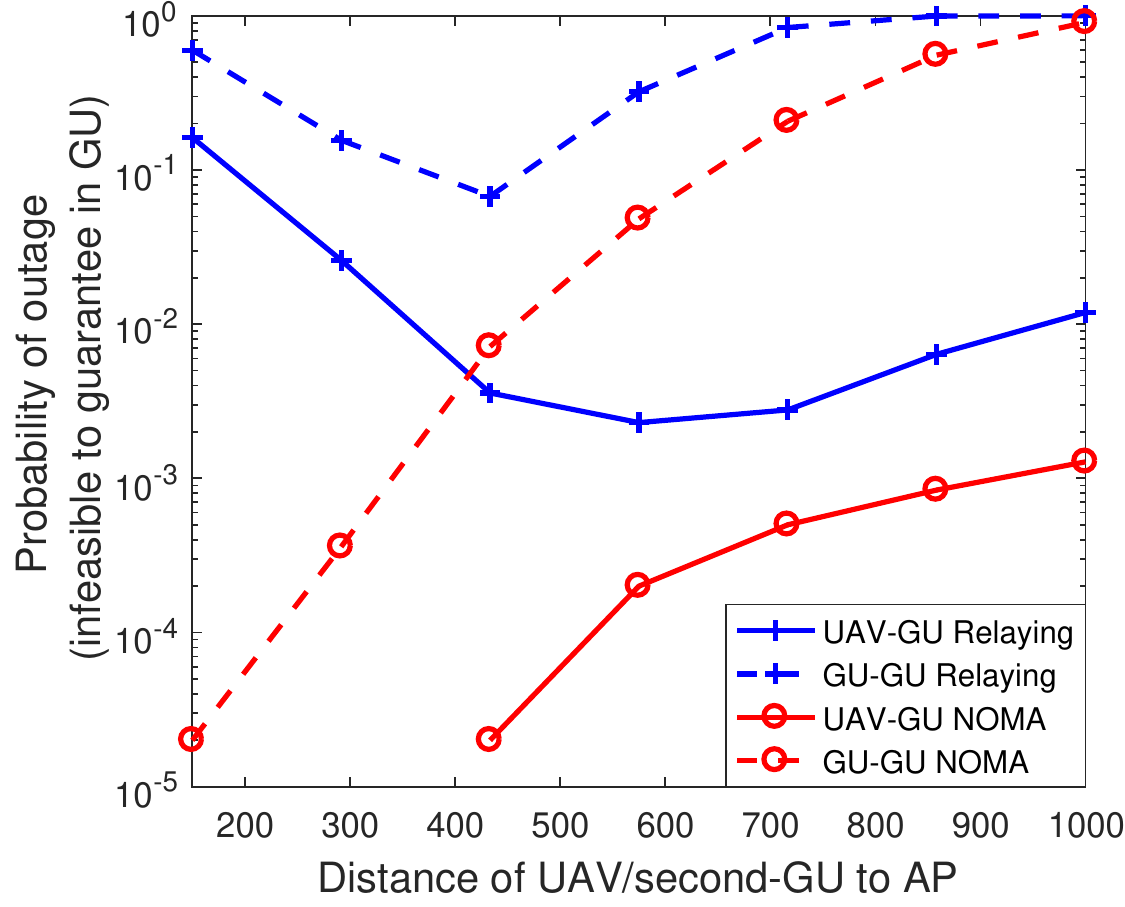}
	\caption{The probability that the first   GU's transmission quality cannot be guaranteed.}
	\label{outage_comp}
\end{figure}
 
Finally, although the designs (for optimal resource allocations for a UAV-GU pair under relaying and NOMA schemes) in this work are motivated by utilizing the more favorable channel conditions (with a high probability of LoS) in links from/to the UAV and utilizing the channel quality difference  between the AP-UAV and AP-GU links,  it should be pointed out that these designs can be also applied  for the scenario with a pair of   GUs. 
In the following, we compare   these two types of user pairs, i.e., UAV-GU and GU-GU, to show how our designs benefit the UAV connected scenario.. We set the locations of the AP and the first (primary) GU to $(0,0,20)$ and $(800,0,0)$, respectively.  In addition, we set the locations of the UAV and the second GU to $(\hat d,0,100)$ and  $(\hat d,0,0)$, i.e.,  letting the UAV and the second GU have the same ground  location/projection,   and vary $\hat d$ from 150 m to 1000 m.   Moreover, we   set the  minimum rate coefficient $\beta = 0.95$. 

The throughput results of the comparison are provided in Fig.~\ref{throughput_comp}.
Clearly, owing to the high probability of LoS, the UAV achieves a significantly higher throughput via our designs  than the  second GU. Moreover, relaying has   better performances for both the UAV-GU and GU-GU cases in comparison to NOMA.

The corresponding outage probability results are shown in  Fig.~\ref{outage_comp}, where the outage probability represents the probability that the first (primary) GU's transmission quality cannot be guaranteed, i.e., the probability that   the UAV or the second GU cannot be additionally serviced  by the AP. 
 It can be observed that the outage probability of the second GU is significantly higher than that of the UAV. In particular, in order to guarantee the first GU's transmission quality, the GU-GU pair is not able to provide a reliable connection for the second GU when this GU is a relatively far from the AP, e.g., when the distance of the second GU to the AP is longer than 600, the outage probability of the second GU becomes higher than $10^{-1}$. On the other hand, the UAV-GU pair introduces a relatively reliable connection to the UAV. Surprisingly, the NOMA scheme   is more reliable for the UAV connection than relaying, which is totally different from the results in  Fig.~\ref{throughput_comp} with respect to the throughput.  In particular, this performance advantage of NOMA with respect to outage probability is more significant when the UAV is   close  to the AP.

Combining the results in Fig.~\ref{throughput_comp} and Fig.~\ref{outage_comp}, we conclude that {the NOMA scheme has a higher probability to provide a connection to the UAV, while this connection generally results in a lower throughput to the UAV than the relaying scheme.
Hence, the NOMA scheme is more suitable for data traffic with small  data packet sizes and requiring a high probability of connection.} 
Moreover, 
 a hybrid scheme, which dynamically shifts between NOMA and relaying schemes, is strongly suggested for the considered UAV-connected ground network. 
 This hybrid scheme definitely achieves   better throughput and outage probability performances than the pure NOMA and relaying schemes, which is more important for UAVs with low-latency high reliability  applications.

\section{Conclusion}

In this paper, we have studied a ground wireless network supporting UAV transmissions, and considered NOMA and relaying schemes. 
The throughputs of UAV and GU are characterized in the IBL regime  and the FBL regime, respectively. 
Following the throughput characterizations, {we have provided optimal resource allocation designs for the NOMA and relaying schemes in both the IBL and FBL regimes. 
In particular, the objective of the designs is to maximize the UAV throughput while guaranteeing the GU's transmission quality with respect to throughput and reliability.}

Via simulations, we have validated our analytical model and evaluated the system performance. {In particular, the results suggest that the NOMA scheme is more preferred if the GU has a good link to the AP (e.g., GU is  within a short distance to the AP), while the relaying scheme demonstrates its advantage when the GU is relatively far form the AP.} 
 Moreover, the relaying scheme is able to create a win-win situation, i.e., improve the performance of the GU while additionally providing a transmission service to the UAV. 
 More importantly and surprisingly, we have observed that the relaying scheme offers a higher throughput to the UAV in the FBL regime than in the IBL regime, due to the fact that in the FBL regime less resources are required to guarantee the GU's transmission quality and therefore the extra resources can be allocated to the transmission to the UAV. {However, it is also observed that (in comparsion to relaying)  the NOMA scheme    provides a   higher UAV throughput by slightly sacrificing the GU's performance and is more reliable in  providing a connection to the UAV than relaying with respect to the outage probability.}
 Hence, a hybrid scheme, which dynamically shifts between NOMA and relaying schemes, is strongly advocated to achieve a high throughput while having a lower outage probability.


%

\end{document}